\documentclass[fleqn,11pt]{wlscirep}
\usepackage[utf8]{inputenc}
\usepackage[T1]{fontenc}
\usepackage{longtable}
\usepackage{tabularx}
\usepackage{multirow}
\usepackage{amsmath, amssymb, amsfonts, amsthm}
\usepackage{bm}
\usepackage{algorithm}
\usepackage{algorithmic}

\newtheorem{theorem}{Theorem}
\newtheorem{lemma}{Lemma}
\newtheorem{proposition}{Proposition}

\title{Random walk sampling in social networks involving private nodes}

\author[1, *]{Kazuki Nakajima}
\author[1]{Kazuyuki Shudo}

\affil[1]{Department of  Mathematical  and Computing Science, Tokyo  Institute  of Technology,  Meguro-ku, Tokyo  152-8552, Japan.}
\affil[*]{Corresponding author: nakajima.k.an@m.titech.ac.jp}

\begin{abstract}
Analysis of social networks with limited data access is challenging for third parties. 
To address this challenge, a number of studies have developed algorithms that estimate properties of social networks via a simple random walk. 
However, most existing algorithms do not assume private nodes that do not publish their neighbors' data when they are queried in empirical social networks.
Here we propose a practical framework for estimating properties via random walk-based sampling in social networks involving private nodes.
First, we develop a sampling algorithm by extending a simple random walk to the case of social networks involving private nodes.
Then, we propose estimators with reduced biases induced by private nodes for the network size, average degree, and density of the node label.
Our results show that the proposed estimators reduce biases induced by private nodes in the existing estimators by up to $92.6\%$ on social network datasets involving private nodes.
\end{abstract}


\begin{document}

\flushbottom
\maketitle
%
%
\thispagestyle{empty}


\section{Introduction}

Online social networks (OSNs) have been studied for the purpose of understanding social characteristics such as human connections and behaviors on a worldwide scale.
A long series of studies on OSNs have analyzed properties of social networks where nodes represent users and edges represent friendships between users in an OSN.
One performs a sampling of graph data for analyses using public interfaces when complete data are not available because of privacy concerns.
Crawling methods (e.g., breadth-first search and random walk) in which one repeatedly traverses a neighbor have been used for sampling graph data in the OSNs where public interfaces to retrieve neighbors of a user are available \cite{ahn2007, mislove2007, wilson2009, gjoka2010, kwak2010, gjoka2011, fukuda2022}.
A common challenge is how to accurately estimate properties using a small number of queries because (i) crawling methods typically induce sampling bias toward high-degree nodes and (ii) public interfaces typically limit the maximum number of queries within a particular time interval.

Re-weighted random walk is a practical framework for unbiased estimation of properties of the OSNs \cite{gjoka2010, gjoka2011}.
In this framework, one first performs a simple random walk on the underlying network (i.e., one repeats to select a neighbor uniformly at random and move to that neighbor) to obtain a sequence of sampled nodes that has the Markov property (i.e., a sampled node depends on only the previous sampled node).
Then, one obtains unbiased estimates of properties by re-weighting each sampled node to correct the sampling bias derived from the Markov property.
A number of studies have developed algorithms that estimate properties of the OSNs using this framework \cite{katzir2011, katzir2013, gjoka2013, dasgupta2014, wang2014, han2016, chen2016, nakajima2020jip, bera2020}.

However, most existing algorithms do not assume private nodes that do not publish their neighbors' data when they are queried in practical scenarios.
A few previous studies reported that private nodes accounted for $27\%$ of all the nodes on Facebook \cite{catanese2011} and $34\%$ of all the nodes on Pokec, which is an OSN in Slovakia \cite{takac2012}.
Private nodes raise practical problems when one attempts to apply existing algorithms to empirical social networks.
First, how do we deal with private nodes to obtain a sample sequence via a simple random walk? 
If a walker visits a private node, one can handle an exception wherein the neighbors' data of the node are not retrievable by jumping to some public user sampled previously. 
However, if one performs such exceptional processes, the sample sequence typically loses the Markov property, which prevents us from obtaining unbiased estimates of proeprties.
There is another serious problem. 
A temporary solution to problems in the sampling phase is to not visit private nodes, as in the case study of random walks on the Facebook graph \cite{gjoka2010, gjoka2011}.
However, if a walker does not traverse private nodes, the conventional framework \cite{gjoka2010, gjoka2011}, which attempts to correct only the sampling bias, is expected to induce biases due to private nodes in estimators.

In this study, we aim to provide a practical framework for estimating properties based on a random walk on social networks involving private nodes.
To this end, we first make three assumptions with respect to private nodes and formalize two models for accessing graph data, called the ideal model and the hidden privacy model.
The assumptions and access models are based on previous studies and our observations on empirical social networks involving private nodes. 
Then, we design a sampling algorithm based on a random walk and develop estimators for the network size (i.e., number of nodes), average degree, and density of the node label (e.g., fraction of nodes with a given label).
Our framework may help to extend random walk-based estimators of properties of a network to the case of social networks involving private nodes.

This study has three main contributions.
First, we develop a sampling algorithm that practically works on social networks involving private nodes (Section \ref{section:4}).
We design a procedure of neighbor selection, which is a fundamental element in the sampling phase via a random walk on the network, and derive the sampling bias of each node induced by the walk.
Then, for each access model, we describe how to calculate the weight for each sampled node, which is essential to correct the sampling bias.
Furthermore, we propose a method to estimate the weight using a much smaller number of queries than the exact calculation method for the hidden privacy model.

Second, we present estimators with reduced biases induced by private nodes for the network size, average degree, and density of the node label (Section \ref{section:5}).
Existing estimators are expected to induce biases due to private nodes because the conventional framework assumes the correction of only sampling bias.
In our framework, we re-weight each sampled node to attempt to correct both the sampling bias and the bias induced by private nodes.
Furthermore, we theoretically show that the proposed estimators have approximately no bias induced by private nodes if all public nodes induce one connected component of the original network.

Third, we validate the theoretical results and effectiveness of the proposed estimators using empirical social network datasets (Section \ref{section:6}).  
We show that the proposed estimators acceptably perform on the two empirical datasets involving private nodes.
Specifically, for the Pokec social network dataset \cite{takac2012}, the proposed estimators reduce biases induced by private nodes in the existing estimators by up to $92.6\%$.
For the Facebook dataset \cite{kurant}, the proposed estimators provide reasonable estimates of the network size, average degree, and cumulative degree distribution of the Facebook graph as of 2010.

\section{Related Work}

Over the last decade, a number of algorithms based on the re-weighted random walk have been developed for estimating properties of the OSNs.
Examples of properties of interest include the network size \cite{katzir2011, katzir2013}, average degree \cite{dasgupta2014, gjoka2010, gjoka2011}, degree distribution \cite{gjoka2010, gjoka2011}, joint degree distribution \cite{gjoka2013}, clustering coefficients \cite{ribeiro2010, katzir2013, bera2020}, motifs and graphlets \cite{wang2014, chen2016, han2016}, and node centrality \cite{nakajima2020jip}.
In this study, we address effects of private nodes on the re-weighted random walk.
Specifically, we first design a sampling algorithm based on a random walk considering private nodes.
Second, we estimators with reduced biases induced by private nodes for the network size, average degree, and density of the node label.
The proposed framework may help us to extend random walk-based estimators of properties of a network, not limited to these three properties, to the case of the OSNs involving private nodes.

Several studies have proposed random walk algorithms to improve the estimation accuracy or the efficiency of the number of queries over a simple random walk \cite{ribeiro2010, lee2012, li2015, nazi2015, zhou2016, li2019, yi2021}.
Ribeiro and Towsley proposed multidimensional random walks, which improve the estimation accuracy in the presence of multiple connected components \cite{ribeiro2010}.
Lee et al. proposed the non-backtracking random walk algorithm, which improves the query efficiency while preserving the Markov property of the sample sequence \cite{lee2012}.
Yi et al. proposed a random walk-based algorithm that reduces biases of estimators using the bootstrapping technique \cite{yi2021}.
These algorithms assume social networks involving no private nodes.
In this study, we extend a simple random walk to the case of social networks involving private nodes.
Based on our work, it is not trivial but possible to extend these improved random walks to the case of social networks involving private nodes.

Re-weighted random walk is a special case of respondent-driven sampling (RDS) \cite{gjoka2010, gjoka2011}.
RDS is a random walk-based sampling method for estimating the proportion of individuals in the hard-to-reach population in social surveys (e.g., the fraction of infected individuals and the fraction of injection drug users) \cite{heckathorn1997, salganik2004, volz2008}.
In the context of the RDS, a private node corresponds to an individual who will not respond to a survey at all.
Such individuals with no response are easily present in practical scenarios \cite{robins2004, goel2010, illenberger2012, gile2015, western2016}.
Several studies numerically investigated the bias of the RDS induced by no-response individuals.
Tomas and Gile numerically showed that the estimator is biased when the response rate changes depending on the degree of the individual and the presence or absence of the infection of the individual \cite{tomas2011}.
Lu et al. numerically investigated the bias when each individual does not respond with a given probability and showed that changes in the probability little affect the bias \cite{lu2012}.
Rocha et al. investigated the effect of the community structure on the bias when each individual does not respond with a given probability \cite{Rocha2017}.
In this study, in the terminology of the RDS, we assume that each individual independently does not respond to a survey at all with a given probability (see Assumption \ref{assumption:2} in Section \ref{section:3.2} for details).
Then, we propose an estimator of the density of the node label in social networks involving private nodes (Section \ref{section:5.4}). 
Note that the density of the node label corresponds to the proportion of individuals with a specific characteristic (e.g., infected individual or drug user) in the context of the RDS.
We theoretically and numerically show that the proposed estimator has little bias induced by private nodes.

Private nodes are regarded as missing graph data in random walk-based estimators because we are not permitted to retrieve their neighbors' data.
In this study, we assume that each node becomes a private node independently at random with a given probability (see Assumption \ref{assumption:2} in Section \ref{section:3.2} for details).
Several studies investigated the effects of completely random missing nodes on the structural properties of a network \cite{albert2000, kossinets2006, huisman2009, smith2013, cohen2000, nakajima2021scirep}.
Albert et al. found the robustness of networks against randomly missing nodes (i.e., if a fraction of nodes are randomly removed from the original network, a large fraction of the remaining nodes induce the largest connected component) \cite{albert2000}.
Kossinets showed that the bias of the average degree between the original network and the remaining largest connected component increases linearly with the proportion of randomly missing nodes \cite{kossinets2006}. 
In this study, we theoretically analyze these biases and design estimators to reduce them under specific assumptions and access models.
There are two important findings in our study compared with these previous studies.
First, although we are not allowed to retrieve the neighbors of a private node by querying the node, we can find the private node in its neighbors that are public nodes in social networks (see Assumption \ref{assumption:1} in Section \ref{section:3.2} for details).
Second, we can reduce the bias induced by private nodes by modifying the weight for each sampled node if each node becomes a private node independently at random with a given probability.

The present work extends the preliminary version \cite{nakajima2020kdd}.
We present the experimental results for all the datasets here.
A major new contribution is to propose a random walk-based estimator of the density of the node label for social networks involving private nodes (Section \ref{section:5.4}) and numerically evaluate the proposed estimator (Sections \ref{section:6.2.1}, \ref{section:6.2.2}, and \ref{section:6.2.3}).
The density of the node label is the leading property of interest in the context of the RDS \cite{heckathorn1997, salganik2004, volz2008} and one of the properties of interest in the context of the estimation in OSNs \cite{ribeiro2010, gjoka2010, gjoka2011}.
We also describe heuristic estimators of the density of the private label (i.e., the fraction of private nodes) in Section \ref{section:5.5}.

\section{Preliminaries} \label{section:3}

\subsection{Definitions and Notations}
We represent a social network as a connected and undirected graph $G = (V, E)$ that consists of a set of nodes $V = \{v_1, ..., v_n\}$ and a set of edges $E$, where $n$ is the number of nodes. 
We denote by $\Gamma(i) = \{v_j\ |\ (v_i, v_j) \in E\}$ a set of neighbors of node $v_i$.
Let $d_i = |\Gamma(i)|$ denote the degree (i.e., the number of neighbors) of node $v_i$ and $D = \sum_{i=1}^n d_i$ denote the sum of degrees. 
We define the average degree of $G$ as $d_{\text{avg}} = D/n$. 
Each node $v_i$ is associated with a set of labels $\mathcal{L}(i)$.
We define the density of node label $l$ as the fraction of nodes with label $l$, i.e., $\rho(l) = \sum_{i=1}^n 1_{\{l \in \mathcal{L}(i)\}}/n$, where $1_{\{cond\}}$ denotes an indicator function that returns 1 if a condition $cond$ holds and 0 otherwise.
For example, when one is interested in the fraction of nodes with a given degree $k$, one sets the indicator function to ensure that it returns 1 if a node has degree $k$ and 0 otherwise.
Each node $v_i$ has a privacy label $l_{\text{pri}}(i) \in \{\text{public}, \text{private}\}$.
Note that $l_{\text{pri}}(i) \in \mathcal{L}(i)$.
We call a node that has a private label a private node and call a node that has a public label a public node.
The set of privacy labels of all the nodes is denoted by $\mathcal{L}_{\text{pri}} = \{l_{\text{pri}}(i)\}_{i=1}^n$. 
 
We refer to connected components induced from public nodes of the original network as {\it public clusters}\footnote{This term is based on Ref.~\cite{albert2000} in which the authors use `clusters' to refer to connected components of a network that remain after a fraction of nodes are removed from the original network.}.
We denote by $C^* = (V^*, E^*)$ the largest public cluster.
Let $n^* = |V^*|$ denote the number of nodes in $C^*$.
We call the neighbors of a node that are public nodes public neighbors of the node.
We denote by $d_i^* = |\{v_j \in V^*\ |\ (v_i,v_j) \in E^* \}|$ the {\it public degree} (i.e., the number of public neighbors) of node $v_i \in V^*$. 
Let $D^* = \sum_{v_i \in V^*} d_i^*$ denote the sum of public degrees. 
We define the average degree of $C^*$ as $d_{\text{avg}}^* = D^*/n^*$. 
We also define the density of node label $l$ of $C^*$ as $\rho^*(l) = \sum_{v_i \in V^*} 1_{\{l \in \mathcal{L}(i)\}}/n^*$.
Figure \ref{fig:1} shows an example of a graph with privacy labels.
There are three public clusters, $C_1$, $C_2$, and $C_3$:
\begin{itemize}
  \item $C_1 = (\{v_1, v_2, v_4, v_5\}, \{(v_1, v_2), (v_1, v_4), (v_4, v_5)\})$
  \item $C_2 = (\{v_8, v_9\}, \{(v_8, v_9)\})$
  \item $C_3 = (\{v_6\}, \{\})$.
\end{itemize}
Note that $C^* = C_1$.
It holds that $n^* = 4, d_1^* = 2, d_2^* = 1, d_4^*=2, d_5^*=1, D^* = 6, d_{\text{avg}}^* = 3/2$.
When one sets the indicator function in $\rho^*(l)$ to ensure that it returns 1 if a given node has degree three (i.e., $d_i=3$) and 0 otherwise, one obtains $\rho^*(l) = 3/4$.

\begin{figure}[t]
        \begin{center}
          \includegraphics[scale=0.4]{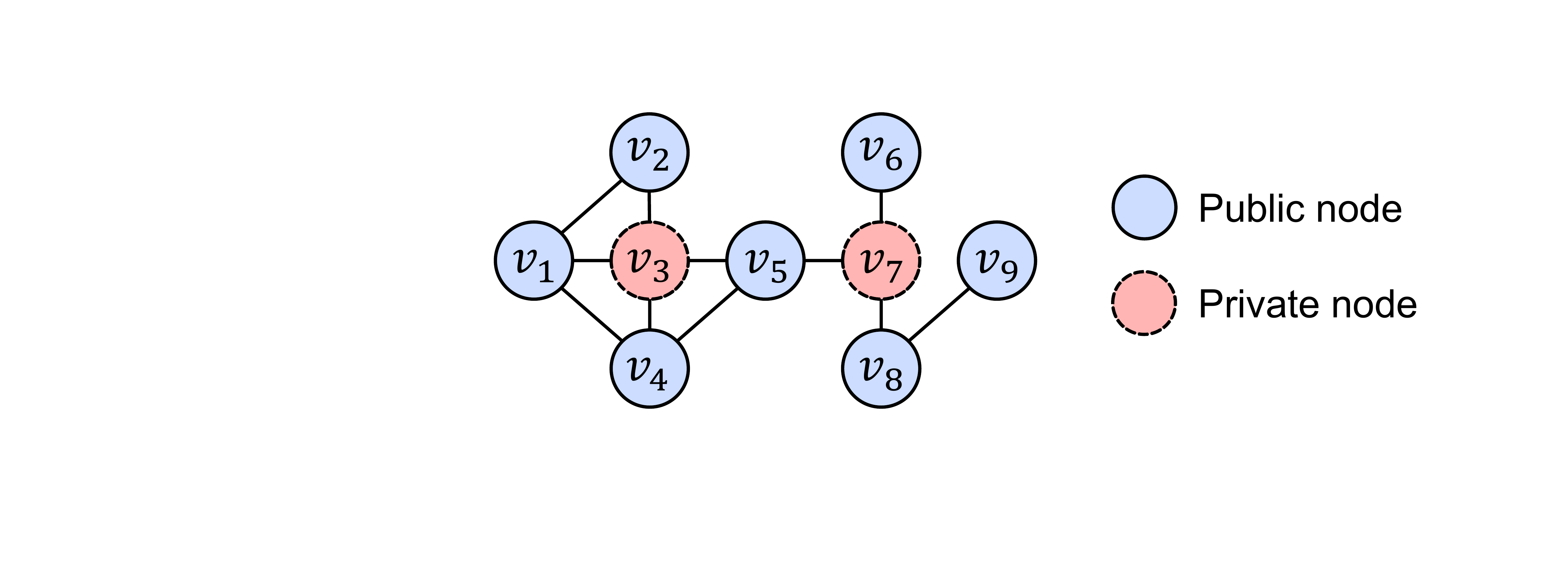}
        \end{center}
      \caption{An example of a graph with privacy labels. Nodes 3 and 7 are private nodes, and all other nodes are public nodes.}
      \label{fig:1}
\end{figure}

\subsection{Assumptions} \label{section:3.2}
We make the following three assumptions.

\begin{enumerate}
\item {\it If we query a public node, the indices of all its neighbors are available. }
This assumption allows us to retrieve a private node from its neighbors that are public nodes.
For example, when one queries private node $v_3$ in Fig. \ref{fig:1}, its neighbors are not retrievable; however, when one queries node $v_5$, all its neighbors, i.e., $v_3$, $v_4$, and $v_7$, are retrievable.
We empirically find that this assumption sufficiently holds in practical scenarios: 
(i) the Facebook graph as of the previous study \cite{gjoka2011} satisfied this assumption;  
(ii) the public interfaces of Twitter as of December 2021 satisfy this assumption \cite{twitter_api_followers, twitter_api_friends}; and
(iii) in the context of social surveys, even if an individual decides not to respond to the survey at all, the individual is encouraged to participate in the survey by its participating neighbors. 
\label{assumption:1}

\item {\it Each node independently becomes a private node with probability $p$ and becomes a public node otherwise, where $0 \leq p < 1$.} 
Private nodes tend to have low degrees under this assumption in empirical social networks.
This is because empirical social networks have heavy-tailed degree distributions \cite{ahn2007, gjoka2011, gjoka2010, kwak2010, mislove2007}.
We validate effectiveness of our estimators designed under this assumption by using social network datasets involving private nodes in Sections \ref{section:6.2.2} and \ref{section:6.2.3}. 
\label{assumption:2}

\item {\it We have access to some arbitrary node that belongs to the largest public cluster to begin our random walk.} 
Private nodes restrict a set of public nodes that a walker is allowed to reach on the network. 
For example, if one selects node $v_5$ as a seed in Fig.~\ref{fig:1}, private node $v_7$ inhibits a walker from reaching public nodes $v_6$, $v_8$, and $v_9$.
Under this assumption, a waller is allowed to traverse the largest public cluster.
We do not consider the number of queries generated for selecting a seed that belongs to the largest public cluster.
We discuss the validity of this assumption in practical scenarios in Section \ref{section:6.2.5}.
\label{assumption:3}
\end{enumerate}

\subsection{Access Models}
We define access models for accessing graph $G$.
We extend a standard access model \cite{chiericetti, gjoka2010, gjoka2011, ribeiro2010} to access models involving private nodes.
Suppose we queried node $v_i$.
If node $v_i$ is a public node, then the neighbors' data of $v_i$ and the set of labels of $v_i$, i.e., $\mathcal{L}(i)$, are available.
If node $v_i$ is a private node, then the neighbors' data of $v_i$ and $\mathcal{L}(i)$ are not available\footnote{We assume that the response is an empty set when one queries a private node.}. 
We consider two models for available neighbors' data of a queried public node $v_i$: the {\it ideal model} and the {\it hidden privacy model}.

In the ideal model, when one queries node $v_i$, the indices and privacy labels of all the neighbors of $v_i$ are available.
For example, when we query a public node $v_4$ in Fig. \ref{fig:1}, we obtain the set $\{(v_1, \text{public}), (v_3, \text{private}), (v_5, \text{public})\}$. 
As empirical example, the Facebook graph as of the previous study \cite{gjoka2010, gjoka2011} corresponds to this access model. 

In the hidden privacy model, when we query node $v_i$, the indices of all the neighbors of $v_i$ are available but their privacy labels are not available.
For example, when we query a public node $v_4$ in Fig. \ref{fig:1}, we obtain the set $\{v_1, v_3, v_5\}$. 
Empirical examples corresponding to this access model include the public interfaces of Twitter as of December 2021 \cite{twitter_api_followers, twitter_api_friends} and social networks in the context of social surveys \cite{heckathorn1997, salganik2004, volz2008}.

\subsection{Markov Chain}
We introduce the basics of a Markov chain for the theoretical analysis of random walk-based estimators.
First, we describe the stationary distribution of a Markov chain, which serves to derive the sampling bias induced by a random walk. 
Let $\bm{P} = (P_{i,j})_{i,j \in S}$ denote the transition probability matrix of a Markov chain on a finite state space $S$. 
If it holds that $\pi_j  = \sum_{i \in S} {\pi_i P_{i,j}}$ for all $j \in S$, a vector $\bm{\pi} = (\pi_i)_{i \in S}$ is the stationary distribution of the chain. 
If a Markov chain is irreducible and aperiodic, then the chain is said to be ergodic (see \cite{levin2017} for formal definitions).
The following theorem holds in regard to the stationary distribution $\bm{\pi}$ of a Markov chain.
\begin{theorem} \label{ergodic}
\cite{levin2017} If a Markov chain is ergodic, the stationary distribution $\bm{\pi}$ uniquely exists.
\end{theorem}

Then, we review the strong law of large numbers for a Markov chain, which ensures that an estimator converges almost surely to its expected value with respect to the stationary distribution \cite{lee2012, chen2016}:

\begin{theorem}\label{SLLN}
\cite{jones, roberts}
Let $\{X_k\}_{k=1}^r$ be an ergodic Markov chain with the stationary distribution $\bm{\pi}$ on a finite state space $S$. 
For any function $f: S \to \mathbb{R}$, the quantity $\sum_{k=1}^r f(X_k)/r$ converges to the expected value with respect to $\bm{\pi}$, i.e., $\mathbb{E}_{\bm{\pi}}[f] \triangleq \sum_{i \in S} \pi_i f(i)$, almost surely as $r \to \infty$ regardless of the initial distribution of the chain.
\end{theorem}

\section{Sampling Algorithm} \label{section:4}
In this section, we design a sampling algorithm based on a random walk considering private nodes in each access model.
First, we describe how to select a neighbor to be traversed from the current node to obtain a sample sequence that has the Markov property.
Then, we derive the sampling bias of each node induced by our random walk.
Finally, we describe a method for calculating the public degree of each sampled node to correct the sampling bias.

\subsection{Neighbor Selection} \label{neighbor}
In a simple random walk, one repeats to select a neighbor uniformly at random and move to that neighbor.
If a walker visits a private node in this method, two problems occur for existing estimators based on a simple random walk.
First, we are not allowed to continue the walk because neighbors of the private node are not retrievable.
Although we can restart the walk from an arbitrary public node previously sampled, the sample sequence loses the Markov property by performing such exception handling. 
Second, it is difficult to correct the sampling bias of private nodes because their degrees and public degrees are unclear. 

We extend a simple random walk to the case of social networks involving private nodes.
We collect a sequence of indices of $r$ sampled nodes, denoted by $({x_1}, {x_2}, \ldots, {x_r})$, as follows. 
We select a seed $v_{x_1} \in C^*$, according to Assumption \ref{assumption:3}.
For the $k$-th sampled node ($k = 1, \ldots, r-1$), we first obtain the set of neighbors of $v_{x_k}$, i.e., $\Gamma(x_k)$, by querying $v_{x_k}$.
Then, we uniformly and randomly select node $u$ from the set $\Gamma(x_k)$.
If $u$ is a public node, the walker moves to $u$ as the next sampled node $v_{x_{k+1}}$, otherwise, we uniformly and randomly select node $u$ from the set $\Gamma(x_k)$ again.
In the ideal model, where the privacy labels of all the neighbors of a queried node are available, we check if a selected neighbor $u$ is public without querying node $u$.
In the hidden privacy model, where the privacy labels of any neighbors of a queried node are not available, we check if $u$ is public by additionally querying $u$.
Note that a walker that is located at $v_{x_k}$ has at least one candidate public neighbor to be traversed for each $k = 1, \ldots, r$ if and only if $v_{x_1}$ belongs to the largest public cluster composed of multiple public nodes.

\begin{figure}[t]
        \begin{center}
          \includegraphics[scale=0.3]{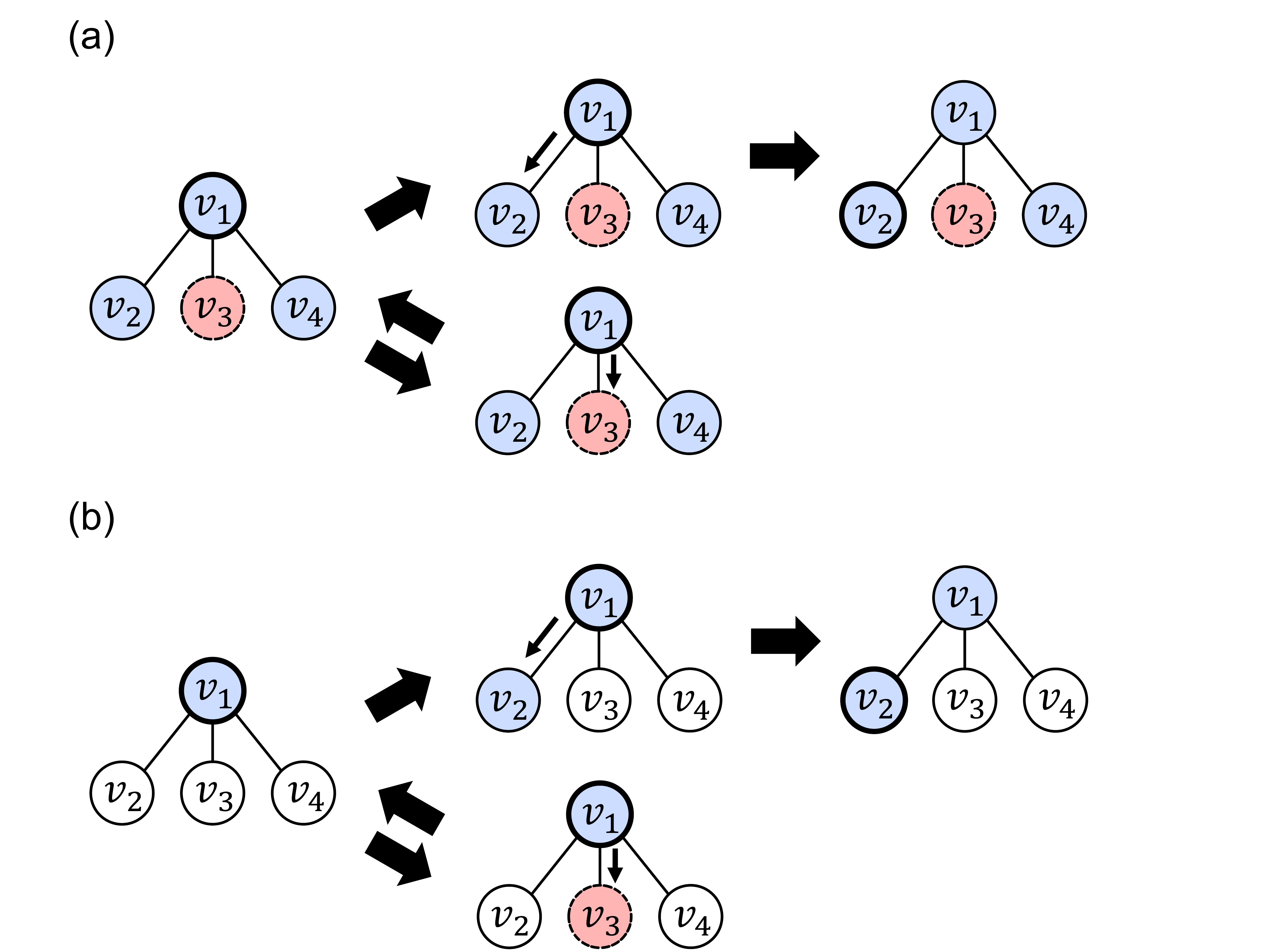}
        \end{center}
      \caption{An example of the procedure of neighbor selection when a walker is located at node $v_1$. A node with a thick line circle represents a node at which a walker is located. (a) Ideal model. (b) Hidden privacy model.}
      \label{fig:2}
\end{figure}

We show in Fig. \ref{fig:2} an example of the procedure of neighbor selection when a walker is located at node $v_1$, where $v_1$ has neighbors, i.e., nodes $v_2$, $v_3$, and $v_4$.
In the case of the ideal model, if one selects neighbor $v_2$ (see the upper side of Fig. \ref{fig:2}(a)), a walker moves to neighbor $v_2$ because node $v_2$ is a public node.
On the other hand, if one selects neighbor $v_3$ (see the lower side of Fig. \ref{fig:2}(a)), one reselects a neighbor uniformly at random because node $v_3$ is a private node.
The same applies to the hidden privacy model (see Fig. \ref{fig:2}(b)).
Note that one needs to check if a selected neighbor is a public node by querying the neighbor in the hidden privacy model.

\subsection{Sampling Bias}
We derive the sampling bias induced by our random walk.
Let the probability that an event $A$ will occur be denoted by $\text{Pr}[A]$. 
We define the distribution induced by the sequence of sampled indices as $\bm{\pi}_{r} = (\text{Pr}[x_r = i])_{i=1}^n$, where $\text{Pr}[x_r = i]$ is the probability that a walker traverses node $v_i$ at $r$-th step.
Note that $\sum_{i=1}^n \text{Pr}[x_r = i] = 1$.
The following lemma indicates that each node that belongs to $C^*$ is sampled in proportion to the public degree via our random walk.

\begin{lemma}\label{lemma:1}
The vector $\bm{\pi}_{r}$ converges to $\bm{\pi} = (p_i)_{i=1}^n$ after many steps of our random walk, where $p_i = 1_{{\{v_i \in V^*\}}} d_i^*/D^*$.
\end{lemma}

\begin{proof}
First, it holds that $\text{Pr}[x_r = i] = 0$ for each node $v_i \in V \backslash V^*$ because our random walk never traverses nodes that do not belong to $C^*$. 
Then, for each node $v_i \in V^*$, we show that $\text{P}r[x_r = i]$ converges to $d_i^*/D^*$ after many steps of our random walk. 
Our random walk has the transition probability matrix $\bm{P} = (P_{i,j})_{v_i, v_j \in V^*}$ defined as  
\begin{align*}
P_{i, j} = 
\begin{cases}
1/d_i^* & \text{if }(v_i, v_j) \in E^*, \\
0 & (\text{otherwise}).
\end{cases}
\end{align*}
The corresponding Markov chain is ergodic because it is equivalent to a simple random walk on $C^*$. 
Note that a simple random walk on a connected graph is ergodic \cite{Lovasz1996, levin2017}.
Therefore, the stationary distribution uniquely exists because of Theorem \ref{ergodic}. 
The vector $(p_i)_{i=1}^n$ satisfies the definition of the stationary distribution. 
The probability $\text{Pr}[x_r = i]$ converges to the stationary distribution after many steps of our random walk.
\end{proof}

\subsection{Calculating the Public Degree of Each Sampled Node} \label{calc_pubd}
We calculate the public degree of each sampled node to correct the sampling bias that is attributable to the public degree.
In the ideal model, we exactly calculate the public degree of each sampled node without additional queries because the privacy labels of all the neighbors of each sampled node are available. 
On the other hand, in the hidden privacy model, one requires a huge number of additional queries to exactly calculate the public degree of each sampled node.

\begin{algorithm}[t]                     
\caption{Our random walk in the hidden privacy model.}      
\label{alg:1}                          
\begin{algorithmic}[1]               
\REQUIRE Seed $v_{x_1} \in C^*$. Sample size $r$.
\ENSURE Sampling list $R$.
\STATE $R \leftarrow$ an empty list.
\FOR{$k=1$ to $r$}
\STATE Query $v_{x_k}$ and obtain the set $\Gamma(x_k)$.
\STATE $d_{x_k} \leftarrow |\Gamma(x_k)|$.
\STATE $\hat{d}_{x_k}^* \leftarrow 0$.
\STATE $R \leftarrow$ append $(x_k,d_{x_k},\hat{d}_{x_k}^*)$.
\IF{$v_{x_k}$ has been visited for the first time}
\STATE $a_{x_k} \leftarrow 0$.
\STATE $b_{x_k} \leftarrow 0$.
\ENDIF
\STATE flag $\leftarrow$ False.
\WHILE{flag is False}
\STATE $u \leftarrow$ a neighbor uniformly and randomly chosen from $\Gamma(x_k)$.
\STATE $b_{x_k} \leftarrow b_{x_k} + 1$.
\IF{$u$ is a public node}
\STATE $v_{x_{k+1}} \leftarrow u$
\STATE $a_{x_k} \leftarrow a_{x_k} + 1$
\STATE flag $\leftarrow$ True.
\ENDIF
\ENDWHILE
\ENDFOR
\FOR{$k=1$ to $r$}
\STATE $\hat{d}_{x_k}^* \leftarrow d_{x_k} \frac{a_{x_k}}{b_{x_k}}$
\ENDFOR
\RETURN{$R$}
\end{algorithmic}
\end{algorithm}

We propose a method to estimate the public degree of each sampled node without additional queries in the hidden privacy model.
The proposed method utilizes the history of neighbor selections generated by our random walk.
Specifically, we record two quantities $a_{x_k}$ and $b_{x_k}$ for each sampled node $v_{x_k}$.
Quantity $a_{x_k}$ is the total number of times a public neighbor of $v_{x_k}$ is successfully selected. 
Quantity $b_{x_k}$ is the total number of times a neighbor of $v_{x_k}$ is selected.
For example, we consider the case of Fig.~\ref{fig:2}(b) in which a walker is located at node $v_1$ in the hidden privacy model.
If one selects neighbor $v_2$, one increases $a_1$ and $b_1$ by one each because $v_2$ is a public neighbor of $v_1$ (see the upper side of Fig. \ref{fig:2}(b)).
On the other hand, if one selects neighbor $v_3$, one increases $b_1$ by one because $v_3$ is not a public neighbor of $v_1$ (see the lower side of Fig. \ref{fig:2}(b)).
After completing our random walk of length $r$, we calculate an estimator, $\hat{d}_{x_k}^*$, of the public degree of each sampled node $v_{x_k}$ as
\begin{align*}
  \hat{d}_{x_k}^* \triangleq d_{x_k} \frac{a_{x_k}}{b_{x_k}}.
\end{align*}
Note that it holds that $b_{x_k}> 0$ for each $k = 1, \ldots, r$ because at least one neighbor selection is performed for each sampled node.
Algorithm \ref{alg:1} shows the pseudocode of our random walk using the proposed method.

We ensure that the estimator $\hat{d}_{x_k}^*$ is an unbiased estimator of the public degree of $v_{x_k}$.
\begin{lemma}\label{lemma:2}
For each sampled node $v_{x_k}$, the estimator, $\hat{d}_{x_k}^*$, converges to the true value, $d_{x_k}^*$, after many steps of our random walk.
\end{lemma}

\begin{proof}
Let $X_{x_k}(l)$ denote a random variable that is equal to 1 if a public neighbor of $v_{x_k}$ is selected at the $l$-th trial of neighbor selections at $v_{x_k}$ and is equal to 0 otherwise, where $l = 1, \ldots, b_{x_k}$ and it holds that $\sum_{l=1}^{b_{x_k}} X_{x_k}(l) = a_{x_k}$. 
It holds that $\text{Pr}[X_{x_k}(l) = 1] = d_{x_k}^*/d_{x_k}$ because $X_{x_k}(l)$ follows a Bernoulli distribution for each $l$. 
Therefore, we have $\mathbb{E}[\hat{d}_{x_k}^*] = d_{x_k} d_{x_k}^*/d_{x_k} = d_{x_k}^*$. 
A sequence of random variables $\{X_{x_k}(l)\}_{l = 1}^{b_{x_k}}$ is drawn from a process of independent and identically distributed trials.
Therefore, the estimator $\hat{d}_{x_k}^*$ converges to its expected value $\mathbb{E}[\hat{d}_{x_k}^*] = d_{x_k}^*$ after many steps of our random walk because of the law of large numbers.
\end{proof}

The estimation accuracy of the public degree for an individual node with a low degree may be poor because a walker does not often traverse the node.
However, the estimation accuracy of the public degree for the set of nodes with a given low degree (e.g., nodes with degree 2) may not be too poor in a network with heavy-tailed degree distribution. 
This is because there are many nodes with a low degree in such networks.
Therefore, when one attempts to estimate properties of a network, estimation accuracy of the public degree of individual nodes may not be a serious issue.
In fact, we numerically find that an estimator of the network size using $\hat{d}_{x_k}^*$ yields almost the same accuracy as that using the exact public degree (see Section 6.2.4 for details).

Then, we theoretically show that our random walk using the proposed method generates much fewer queries than that using the exact calculation method (i.e., one queries all the neighbors of each sampled node).
For simplicity, we do not consider this saving of the neighbors' data of nodes queried once in the theoretical analysis.
We denote by $Q(k)$ the number of queries generated by the exact method at the $k$-th sampled node $v_{x_k}$.
We denote by $Q'(k)$ the number of queries generated by the proposed method at the $k$-th sampled node. 
Let $Q = (\sum_{k=1}^r Q(k))/r$ denote the ratio of the number of queries using the exact method to the sample size.
Let $Q' = (\sum_{k=1}^r Q'(k))/r$ denote the ratio of the number of queries using the proposed method to the sample size.
We have the following lemma.

\begin{lemma}\label{lemma:3}
The expected value of $Q$ with respect to $\bm{\pi}$ is given by 
\begin{align*}
\mathbb{E}_{\bm{\pi}}[Q] = \frac{1}{D^*} \sum_{v_i \in V^*} d_i^* d_i.
\end{align*}
The expected value of $Q'$ with respect to $\bm{\pi}$ is given by 
\begin{align*}
\mathbb{E}_{\bm{\pi}}[Q'] = \frac{1}{D^*} \sum_{v_i \in V^*} d_i.
\end{align*}
\end{lemma}

\begin{proof}
It holds that $Q(k) = d_{x_k}$ because one queries all the neighbors of $v_{x_k}$ in the exact method.  
Therefore, we have 
\begin{align*}
\mathbb{E}_{\bm{\pi}}[Q]
= \mathbb{E}_{\bm{\pi}}[Q(k)] 
= \sum_{v_i \in V^*} {\frac{d_i^*}{D^*}} \mathbb{E}[Q(k) | x_k = i] 
= \frac{1}{D^*} \sum_{v_i \in V^*} d_i^* d_i.
\end{align*}
The first equation holds because of the linearity of expectation. 
The second equation holds because of the law of total expectation and Lemma \ref{lemma:1}.

The quantity $Q'(k)$ follows the geometric distribution with success probability $d_{x_k}^*/d_{x_k}$ because we repeatedly query a neighbor of $v_{x_k}$ uniformly at random until a public neighbor of $v_{x_k}$ is firstly selected. 
Therefore, it holds true that $\mathbb{E}[Q'(k)] = d_{x_k}/d_{x_k}^*$. 
Then, we have
\begin{align*}
\mathbb{E}_{\bm{\pi}}[Q'] 
= \mathbb{E}_{\bm{\pi}}[Q'(k)] 
= \sum_{v_i \in V^*} {\frac{d_i^*}{D^*}} \frac{d_i}{d_i^*} 
= \frac{1}{D^*} \sum_{v_i \in V^*} d_i.
\end{align*}
\end{proof}

Intuitively, $\sum_{v_i \in V^*} d_i^*$ is the order of $\sum_{v_i \in V} d_i$ and $\sum_{v_i \in V^*} d_i^* d_i$ is the order of $\sum_{v_i \in V} d_i^2$.
The sum of squares of degrees is much larger than the sum of degrees for a large-scale network with a heavy-tailed degree distribution \cite{newman_networks}.
Therefore, Lemma \ref{lemma:3} indicates that the proposed method generates much fewer queries than the exact method.

\section{Estimators} \label{section:5}

\subsection{Overview}
In the conventional framework \cite{gjoka2010, gjoka2011}, one re-weights each sampled node using its public degree to correct the sampling bias. 
Therefore, existing estimators converge to the quantities of the largest public cluster.
When the original network comprises only public nodes, as assumed in the conventional framework \cite{gjoka2010, gjoka2011}, the expected values of estimators are equal to the quantities of the original network. 
However, when the original network involves private nodes, existing estimators are expected to induce biases due to private nodes.

It is not trivial to reduce biases induced by private nodes because the exact probability $p$ or the proportion of private nodes are unknown to third parties.
Furthermore, it is difficult to apply existing methods for estimating the probability $p$ or the proportion of private nodes based on the privacy labels of sampled nodes (e.g., Section 3.C.3 in Ref. \cite{gjoka2011}, Section 3.2 in Ref. \cite{katzir2011}, and Section 4.2.3 in Ref. \cite{ribeiro2010}).
This is because private nodes are not included in the sample sequence.

We propose estimators with reduced biases induced by private nodes for the network size, average degree, and density of the node label. 
We re-weight each sampled node using both its degree and its public degree to reduce biases induced by private nodes under Assumption \ref{assumption:2}. 
We theoretically show that the proposed estimators have approximately no bias induced by private nodes if all public nodes belong to the largest public cluster.
In the following, for each of the three properties, we first introduce the existing estimator and then describe our estimator.
Then, we describe heuristic estimators for the fraction of private nodes, which combine the existing and proposed estimators for the network size and average degree each.

\subsection{Network Size}

\subsubsection{Existing estimator}
The node collision estimator is effective for estimating the network size \cite{katzir2013,katzir2011}. 
In this estimator, one counts the number of collisions in the indices of pairs of the sampled nodes whose ordinal numbers in the sample sequence are far away. 
Such pairs of sampled nodes are regarded as being sampled independently of each other from the stationary distribution \cite{katzir2013, katzir2011}.

Formally, the existing estimator of the network size is defined as follows.
Let $I = \{(k,l)\ |\ m \leq |k-l| \land 1 \leq k,\ l \leq r\}$ denote the set of integer pairs that are between $1$ and $r$ and at least a threshold $m$ away. 
We set $m = 0.025 r$, as in the previous study \cite{katzir2013}.
Let $\phi_{k,l}$ denote a variable that is equal to 1 if the indices of $k$-th and $l$-th sampled public nodes are the same, i.e., $x_k = x_l$ (this is called a collision) and is equal to 0 otherwise. 
One defines the average of the number of collisions $\Phi_{\text{size}}$, the average of the weights to correct the sampling bias $\Psi_{\text{size}}$, and a size estimate $\hat{n}$ as
\begin{align*}
\Phi_{\text{size}} = \frac{1}{|I|}\sum_{(k,l) \in I} \phi_{k,l},\ \ \ 
\Psi_{\text{size}} = \frac{1}{|I|}\sum_{(k,l) \in I} \frac{d_{x_k}^*}{d_{x_l}^*},\ \ \ 
\hat{n} \triangleq \frac{\Psi_{\text{size}}}{\Phi_{\text{size}}}.
\end{align*}

We have the following lemma, which is extended from the results shown in previous studies \cite{katzir2013,katzir2011} which assume that the original network involves no private nodes. 

\begin{lemma} \label{lemma:4}
The estimator $\hat{n}$ asymptotically converges to the size of the largest public cluster, i.e., $n^*$, after many steps of our random walk.
\end{lemma}

\begin{proof}
First, we calculate the expected value with respect to $\bm{\pi}$ of $\Phi_{\text{size}}$:
\begin{align*}
\mathbb{E}_{\bm{\pi}}\left[\Phi_{\text{size}}\right] 
&= \mathbb{E}_{\bm{\pi}}[\phi_{k,l}] 
= \sum_{v_i \in V^*} {\left(\frac{d_i^*}{D^*}\right)^2}.
\end{align*}
The first equation holds because of the linearity of expectation. 
The second equation holds because $v_{x_k}$ and $v_{x_l}$ are sampled independently of each other from the stationary distribution. 
Then, we calculate the expected value with respect to $\bm{\pi}$ of $\Psi_{\text{size}}$:
\begin{align*}
\mathbb{E}_{\bm{\pi}}\left[\Psi_{\text{size}}\right] 
= \mathbb{E}_{\bm{\pi}}\left[\frac{d_{x_k}^*}{d_{x_l}^*}\right] 
= \sum_{v_i \in V^*} {\sum_{v_j \in V^*} {\frac{d_i^*}{d_j^*}\frac{d_j^*}{D^*}\frac{d_i^*}{D^*}}} 
= n^* \sum_{v_i \in V^*} {\left(\frac{d_i^*}{D^*}\right)^2}.
\end{align*}
Quantities $\Phi_{\text{size}}$ and $\Psi_{\text{size}}$ intuitively converge to their respective expected values with respect to $\bm{\pi}$ after many steps of our random walk.
Therefore, we conclude $\hat{n}$ asymptotically converges to $n^*$.
\end{proof}

Quantity $n^*$ depends on a set of privacy labels of the nodes, $\mathcal{L}_{\text{pri}}$.
Therefore, to quantify the bias induced by private nodes of the expected value $n^*$, we attempt to derive the expected value of $n^*$ given a set of privacy labels, $\mathcal{L}_{\text{pri}}$, under Assumption \ref{assumption:2}.
Let $\mathbb{E}_{\text{pri}}[X]$ denote the expected value of a random variable $X$ with respect to $\mathcal{L}_{\text{pri}}$ under Assumption \ref{assumption:2}.
To simplify the calculation, we approximate $\mathbb{E}_{\text{pri}}[n^*]$ under the condition that all the public nodes belong to the largest public cluster. 
Under this condition, it holds that $\text{Pr}[v_i \in V^*] = 1-p$ because of Assumption \ref{assumption:2}.

The following lemma holds regarding the expected value of the existing estimator.
\begin{lemma} \label{lemma:5}
If all the public nodes belong to the largest public cluster, we have
\begin{align*}
  \mathbb{E}_{\text{pri}}[n^*] = (1-p)n.
\end{align*}
\end{lemma}

\begin{proof}
We define a random variable $X_{\text{size}}(i) = 1_{{\{v_i \in V^*\}}}$ for each node $v_i \in V$.
Then, it holds that $n^* = \sum_{v_i \in V} X_{\text{size}}(i)$. 
The expected value of $n^*$ with respect to $\mathcal{L}_{\text{pri}}$ under the given condition is given by
\begin{align*}
  \mathbb{E}_{\text{pri}}[n^*] &= \sum_{v_i \in V} \mathbb{E}_{\text{pri}}[X_{\text{size}} (i)] = \sum_{v_i \in V} \text{Pr}[v_i \in V^*] = (1-p)n.
\end{align*}
The first equation holds based on the linearity of expectation.
The second equation holds based on the law of total expectation.
\end{proof}

Lemma \ref{lemma:5} implies that the expected value of the existing estimator has the bias $1-p$.

\subsubsection{Proposed estimator} \label{section:5.2.2}
We modify the weight for each pair of sampled nodes $(v_{x_k}, v_{x_l})$ such that $(k, l) \in I$ to reduce the bias of the expected value induced by private nodes. 
Specifically, we define the average of the modified weights $\Psi'_{\text{size}}$ and the proposed estimator $\hat{n}'$ as follows:
\begin{align*}
\Psi'_{\text{size}} = \frac{1}{|I|}\sum_{(k,l) \in I} \frac{d_{x_k}}{d_{x_l}^*},\ \ \ 
\hat{n}' \triangleq \frac{\Psi'_{\text{size}}}{\Phi_{\text{size}}}.
\end{align*}

The following lemma holds in regard to the expected value of the proposed estimator. 
\begin{lemma} \label{lemma:6}
The estimator $\hat{n}'$ asymptotically converges to 
\begin{align*}
\tilde{n} = n^* \frac{\sum_{v_i \in V^*} d_i^* d_i}{\sum_{v_i \in V^*} (d_i^*)^2}
\end{align*}
after many steps of our random walk.
\end{lemma}

\begin{proof}
As with the proof of the Lemma \ref{lemma:4}, we have
\begin{align*}
\mathbb{E}_{\bm{\pi}}\left[\Psi'_{\text{size}}\right] 
= \mathbb{E}_{\bm{\pi}}\left[\frac{d_{x_k}}{d_{x_l}^*}\right] 
= n^* \sum_{v_i \in V^*} {\frac{d_i^* d_i}{(D^*)^2}}.
\end{align*}
Quantities $\Phi_{\text{size}}$ and $\Psi'_{\text{size}}$ intuitively converge to the respective expected values after many steps of our random walk.
Therefore, $\hat{n}'$ asymptotically converges to $\tilde{n}$.
\end{proof}

When the original network, $G$, involves no private nodes, the following proposition regarding each estimator and each expected value holds.

\begin{proposition} \label{proposition:1}
When the original network, $G$, involves no private nodes, two estimators, $\hat{n}$ and $\hat{n}'$, are equal, and two expected values, $n^*$ and $\tilde{n}$, are equal to the true quantity, $n$.
\end{proposition}

\begin{proof}
When $G$ involves no private nodes, it holds that $V^* = V$ and $d_i^* = d_i$ for each node $v_i \in V^*$.
Thus, it holds that $\hat{n} = \hat{n}'$ because of the definitions of the estimators. 
It also holds that $n^* = n$ and $\tilde{n} = n$ because of Lemmas \ref{lemma:4} and \ref{lemma:6}.
\end{proof}

We show that the expected value $\tilde{n}$ of the proposed estimator reduces the bias induced by private nodes compared with the existing estimator.
First, we derive the expected value with respect to the set $\mathcal{L}_{\text{pri}}$ of the public degree of a public node:

\begin{lemma}\label{pubd_exp}
For any public node $v_i \in V^*$, we have
\begin{align*}
  \mathbb{E}_{\text{pri}}[d_i^*] &= (1 - p)d_i, \\
  \mathbb{E}_{\text{pri}}\left[(d_i^*)^2 \right] &= (1 - p)d_i[(1-p)d_i+p].
\end{align*}
\end{lemma}

\begin{proof}
The public degree $d_i^*$ follows the binomial distribution with parameters of the degree $d_i$ and $1-p$ regarding the set $\mathcal{L}_{\text{pri}}$ because each neighbor of node $v_i$ independently becomes public with the probability $1-p$ under Assumption \ref{assumption:2}.
\end{proof}

Then, we approximate the expected value of $\tilde{n}$ with respect to $\mathcal{L}_{\text{pri}}$ as a product of each expected value with respect to $\mathcal{L}_{\text{pri}}$ of each quantity in the denominator and numerator of $\tilde{n}$.
\begin{theorem}\label{theorem:3}
If all the public nodes belong to the largest public cluster, we have
\begin{align*}
\mathbb{E}_{\text{pri}}[\tilde{n}] 
\approx \frac{\mathbb{E}_{\text{pri}}[n^* ] \mathbb{E}_{\text{pri}}[\sum_{v_i \in V^*} d_i^* d_i]}{\mathbb{E}_{\text{pri}}[\sum_{v_i \in V^*} (d_i^*)^2]} 
= \alpha_p n,
\end{align*}
where 
\begin{align}
\alpha_p = \frac{(1-p) \sum_{v_i \in V} {(d_i)^2}}{\sum_{v_i \in V} d_i [(1-p)d_i + p]}. \label{eq:3}
\end{align}
\end{theorem}

\begin{proof}
We define a random variables $X_{\text{size}} (i) = d_i^* d_i 1_{{\{v_i \in V^*\}}}$ and $Y_{\text{size}} (i) = (d_i^*)^2 1_{{\{v_i \in V^*\}}}$ for each node $v_i \in V$. 
Let $X_{\text{size}} = \sum_{v_i \in V^*} d_i^* d_i$ and $Y_{\text{size}} = \sum_{v_i \in V^*} (d_i^*)^2$. 
It holds that $X_{\text{size}} = \sum_{v_i \in V} X_{\text{size}} (i)$ and $Y_{\text{size}} = \sum_{v_i \in V} Y_{\text{size}} (i)$. 
We obtain the expected value of $X_{\text{size}}$ with respect to $\mathcal{L}_{\text{pri}}$:
\begin{align*}
  \mathbb{E}_{\text{pri}}[X_{\text{size}}] 
  = \sum_{v_i \in V}\text{Pr}[v_i \in V^*]\mathbb{E}_{\text{pri}}[d_i^* d_i] 
  = (1-p)^2 \sum_{v_i \in V}(d_i)^2.
\end{align*}
The second equation holds because of Lemma \ref{pubd_exp}. 
We note that the degree $d_i$ is constant with respect to $\mathcal{L}_{\text{pri}}$. 
Similarly, the expected value of $Y_{\text{size}}$ with respect to $\mathcal{L}_{\text{pri}}$ is obtained as follows:
\begin{align*}
  \mathbb{E}_{\text{pri}} [Y_{\text{size}}] 
  = \sum_{v_i \in V}\text{Pr}[v_i \in V^*]\mathbb{E}_{\text{pri}}[(d_i^*)^2]
  = (1-p)^2 \sum_{v_i \in V} d_i [(1-p)d_i + p].
\end{align*}
Theorem \ref{theorem:3} holds because of the above equations and Lemma \ref{lemma:5}.
\end{proof}

We empirically find that the coefficient $\alpha_p$ is almost equal to 1 for various values of $p$ in different social networks (see Section \ref{section:6.1.2} for details).
This is because the sum of squares of degrees $\sum_{v_i \in V} {(d_i)^2}$ is considerably larger than the sum of degrees $\sum_{v_i \in V} d_i$ in large-scale networks with heavy-tailed degree distributions \cite{newman_networks}.
Therefore, Theorem \ref{theorem:3} indicates that the expected value of the proposed estimator has approximately no bias with respect to $\mathcal{L}_{\text{pri}}$ if all the public nodes belong to the largest public cluster.

In practice, it rarely holds true that all public nodes belong to the largest public cluster of a large-scale social network. 
Therefore, the expected values of estimators have biases induced by public nodes that do not belong to $C^*$.
Nevertheless, we numerically find that the proposed estimators have smaller biases induced by private nodes than the existing estimators in empirical social networks (see Section \ref{section:6.2.1} for details).

\subsection{Average Degree}
\subsubsection{Existing estimator}
An existing estimator of the average degree \cite{gjoka2010, gjoka2011, dasgupta2014}, denoted by $\hat{d}_{\text{avg}}$, is defined as
\begin{align*}
\Phi_{\text{avg}} = \frac{1}{r} \sum_{k=1}^r {\frac{1}{d_{x_k}^*}}, \ \ \ 
\hat{d}_{\text{avg}} \triangleq \frac{1}{\Phi_{\text{avg}}}.
\end{align*}
We have the following lemma derived from the previous study \cite{dasgupta2014} which assumes that the original network involves no private nodes.

\begin{lemma} \label{lemma:8}
The estimator $\hat{d}_{\text{avg}}$ converges to the average degree of the largest public cluster, i.e., $d_{\text{avg}}^*$, after many steps of our random walk.
\end{lemma}

\begin{proof}
We calculate the expected value of $\Phi_{\text{avg}}$ with respect to $\bm{\pi}$ as follows:
\begin{align*}
\mathbb{E}_{\bm{\pi}}\left[\Phi_{\text{avg}}\right] 
= \mathbb{E}_{\bm{\pi}}\left[\frac{1}{d_{x_k}^*}\right] 
= \sum_{v_i \in V^*} \frac{d_i^*}{D^*} \frac{1}{d_i^*} = \frac{1}{d_{\text{avg}}^*}.
\end{align*}
The quantity $\Phi_{\text{avg}}$ converges to the expected value after many steps because of Theorem \ref{SLLN}.
Therefore, $\hat{d}_{\text{avg}} = 1/\Phi_{\text{avg}}$ converges to $d_{\text{avg}}^*$ after many steps of our random walk.
\end{proof}

Then, we quantify the bias of the expected value of the existing estimator induced by private nodes.
We approximate the expected value of $d_{\text{avg}}^*$ with respect to $\mathcal{L}_{\text{pri}}$ as a product of the expected value with respect to $\mathcal{L}_{\text{pri}}$ of each quantity in the denominator and numerator of $d_{\text{avg}}^*$.
\begin{lemma}\label{lemma:9}
If all the public nodes belong to the largest public cluster, we have
\begin{align*}
  \mathbb{E}_{\text{pri}}[d_{\text{avg}}^*] \approx \frac{\mathbb{E}_{\text{pri}}[D^*]}{\mathbb{E}_{\text{pri}}[n^*]} = (1-p)d_{\text{avg}}.
\end{align*}
\end{lemma}

\begin{proof}
We define a random variable $X_{\text{avg}} (i) = d_i^* 1_{{\{v_i \in V^*\}}}$ for each node $v_i \in V$. 
It holds that $D^* = \sum_{v_i \in V} X_{\text{avg}}(i)$. 
The expected value with respect to $\mathcal{L}_{\text{pri}}$ of $D^*$ is derived as 
\begin{align*}
  \mathbb{E}_{\text{pri}}[D^*] 
  = \sum_{v_i \in V}\text{Pr}[v_i \in V^*]\mathbb{E}_{\text{pri}}[d_i^*] 
  = (1-p)^2 D.
\end{align*}
Consequently, it follows that Lemma \ref{lemma:9} holds because of the above equation and Lemma \ref{lemma:5}.
\end{proof}

Lemma \ref{lemma:9} implies that the expected value of the existing estimator has the bias $1-p$.

\subsubsection{Proposed estimator}
We modify the weight for each sampled node to reduce the bias of the expected value induced by private nodes.
We define the average of the modified weights $\Psi'_{\text{size}}$ and the proposed estimator $\hat{d}_{\text{avg}}'$ as follows:
\begin{align*}
\Phi_{\text{avg}}' = \frac{1}{r} \sum_{k=1}^r {\frac{1}{d_{x_k}}}, \ \ \ 
\hat{d}_{\text{avg}}' \triangleq \frac{1}{\Phi_{\text{avg}}'}.
\end{align*}

We have the following lemma regarding the proposed estimator.
\begin{lemma}\label{lemma:10}
The estimator $\hat{d}_{\text{avg}}'$ converges to
\begin{align*}
\tilde{d}_{\text{avg}} = \frac{D^*}{\sum_{v_i \in V^*} d_i^*/d_i}
\end{align*}
after many steps of our random walk.
\end{lemma}

\begin{proof}
We calculate the expected value of $\Phi_{\text{avg}}'$ with respect to $\bm{\pi}$ as follows:
\begin{align*}
\mathbb{E}_{\bm{\pi}}[\Phi_{\text{avg}}'] 
= \mathbb{E}_{\bm{\pi}}\left[\frac{1}{d_{x_k}}\right] 
= \sum_{v_i \in V^*} \frac{d_i^*}{D^*} \frac{1}{d_i} 
= \frac{1}{\tilde{d}_{\text{avg}}}.
\end{align*}
The quantity $\Phi_{\text{avg}}'$ converges to the expected value because of Theorem \ref{SLLN}, and hence, the estimator $\hat{d}_{\text{avg}}'$ converges to $\tilde{d}_{\text{avg}}$ after many steps of our random walk.
\end{proof}

The following proposition holds as well as Proposition \ref{proposition:1}.
\begin{proposition}\label{proposition:2}
When the original network, $G$, involves no private nodes, two estimators, $\hat{d}_{\text{avg}}$ and $\hat{d}_{\text{avg}}'$, are equal, and two expected values, $d_{\text{avg}}^*$ and $\tilde{d}_{\text{avg}}$, are equal to the true quantity, $d_{\text{avg}}$.
\end{proposition}

Finally, we have the following theorem. 

\begin{theorem}\label{theorem:4}
If all the public nodes belong to the largest public cluster, we have
\begin{align*}
  \mathbb{E}_{\text{pri}}[\tilde{d}_{\text{avg}}] \approx 
  \frac{\mathbb{E}_{\text{pri}}[D^*]}{\mathbb{E}_{\text{pri}} \left[\sum_{v_i \in V^*} d_i^*/d_i \right]} 
  = d_{\text{avg}}.
\end{align*}
\end{theorem}

\begin{proof}
We define a random variable $\tilde{X}_{\text{avg}} (i) = 1_{{\{v_i \in V^*\}}} d_i^* / d_i$ for each node $v_i \in V$. 
Let $\tilde{X}_{\text{avg}} = \sum_{v_i \in V^*} d_i^*/d_i$. 
It holds that $\tilde{X}_{\text{avg}} = \sum_{v_i \in V} \tilde{X}_{\text{avg}} (i)$. 
We obtain the expected value of $\tilde{X}_{\text{avg}}$ with respect to $\mathcal{L}_{\text{pri}}$:
\begin{align*}
  \mathbb{E}_{\text{pri}}[\tilde{X}_{\text{avg}}] 
  = \sum_{v_i \in V}\text{Pr}[v_i \in V^*]\mathbb{E}_{\text{pri}}\left[\frac{d_i^*}{d_i}\right] 
  = (1-p)^2 n.
\end{align*}
Theorem \ref{theorem:4} holds because of the equation of $\mathbb{E}_{\text{pri}}[D^*]=(1-p)^2 D$ and the above equation.
\end{proof}

\subsection{Node's Label Density} \label{section:5.4}
\subsubsection{Existing estimator}
An existing estimator of the density of node label $l$, denoted by $\hat{\rho}(l)$, is defined as follows \cite{gjoka2010, gjoka2011, ribeiro2010}:
\begin{align*}
\Phi_{\text{label}} = \frac{1}{r} \sum_{k=1}^r \frac{1_{{\{l \in \mathcal{L}(x_k)\}}}}{d_{x_k}^*},\ \ \ 
\hat{\rho}(l) \triangleq \frac{\Phi_{\text{label}}}{\Phi_{\text{avg}}}.
\end{align*}
Note that the label of interest $l$ of the node does not include the privacy label of the node because the sample sequence contains only public nodes.

We have the following lemma derived from the previous study \cite{ribeiro2010} which assumes that the original network involves no private nodes.

\begin{lemma}\label{lemma:11}
The estimator $\hat{\rho}(l)$ converges to the density of node label $l$ of the largest public cluster, i.e., $\rho^*(l)$, after many steps of our random walk.
\end{lemma}

\begin{proof}
We calculate the expected value of $\Phi_{\text{label}}$ with respect to $\bm{\pi}$ as follows:
\begin{align*}
\mathbb{E}_{\bm{\pi}}\left[\Phi_{\text{label}}\right] 
= \mathbb{E}_{\bm{\pi}}\left[\frac{1_{{\{l \in \mathcal{L}(x_k)\}}}}{d_{x_k}^*}\right] 
= \sum_{v_i \in V^*} \frac{d_i^*}{D^*} \frac{1_{{\{l \in \mathcal{L}(i)\}}}}{d_{i}^*}
= \frac{1}{D^*} \sum_{v_i \in V^*} 1_{{\{l \in \mathcal{L}(i)\}}}.
\end{align*}
The quantity $\Phi_{\text{label}}$ converges to the expected value $\mathbb{E}_{\bm{\pi}}\left[\Phi_{\text{label}}\right] = \sum_{v_i \in V^*} 1_{{\{l \in \mathcal{L}(i)\}}}/D^*$ after many steps of our random walk because of Theorem \ref{SLLN}.
The quantity $\Phi_{\text{avg}}$ also converges to the expected value $\mathbb{E}_{\bm{\pi}}\left[\Phi_{\text{avg}}\right] = n^*/D^*$ after many steps of our random walk (see the proof of Lemma \ref{lemma:8}).
Therefore, the estimator $\hat{\rho}(l)$ converges to $\rho^*(l) = \sum_{v_i \in V^*} 1_{{\{l \in \mathcal{L}(i)\}}}/n^*$ after many steps of our random walk.
\end{proof}

Then, we quantify the bias of the expected value $\rho^*(l)$ of the existing estimator.

\begin{lemma}\label{lemma:12}
If all the public nodes belong to the largest public cluster, we have
\begin{align*}
  \mathbb{E}_{\text{pri}}[\rho^*(l)] \approx \frac{\mathbb{E}_{\text{pri}}[\sum_{v_i \in V^*} 1_{{\{l \in \mathcal{L}(i)\}}}]}{\mathbb{E}_{\text{pri}}[n^*]} = \rho(l).
\end{align*}
\end{lemma}

\begin{proof}
We define a random variable $X_{\text{label}}(i) = 1_{{\{v_i \in V^* \land l \in \mathcal{L}(i)\}}}$ for each node $v_i \in V$. 
Let $X_{\text{label}} = \sum_{v_i \in V^*} 1_{{\{l \in \mathcal{L}(i)\}}}$.
It holds that $X_{\text{label}} = \sum_{v_i \in V} X_{\text{label}}(i)$. 
The expected value with respect to $\mathcal{L}_{\text{pri}}$ of $X_{\text{label}}$ is derived as
\begin{align*}
  \mathbb{E}_{\text{pri}}\left[X_{\text{label}}\right] 
  = \sum_{v_i \in V} \text{Pr}[v_i \in V^*] \mathbb{E}_{\text{pri}}[1_{{\{l \in \mathcal{L}(i)\}}}]
  = (1-p) \sum_{v_i \in V} 1_{{\{l \in \mathcal{L}(i)\}}}.
\end{align*}
Note that the indicator function $1_{{\{l \in \mathcal{L}(i)\}}}$ is constant with respect to the set $\mathcal{L}_{\text{pri}}$. 
Consequently, it follows that Lemma \ref{lemma:12} holds because of the above equation and Lemma \ref{lemma:5}.
\end{proof}

In contrast to the cases of the network size and average degree, Lemma \ref{lemma:12} implies that the existing estimator $\hat{\rho}(l)$ has approximately no bias with respect to the set $\mathcal{L}_{\text{pri}}$ if all public nodes belong to the largest public cluster.
However, we empirically find that our estimator presented in the following further reduces the bias induced by private nodes.

\subsubsection{Proposed estimator}
We modify the weight for each sampled node to reduce the bias of the expected value induced by private nodes.
We define the average of the modified weights $\Phi_{\text{label}}'$ and the proposed estimator $\hat{\rho}'(l)$ as follows:
\begin{align*}
\Phi_{\text{label}}' = \frac{1}{r} \sum_{k=1}^r \frac{1_{{\{l \in \mathcal{L}(x_k)\}}}}{d_{x_k}},\ \ \ 
\hat{\rho}'(l) \triangleq \frac{\Phi_{\text{label}}'}{\Phi'_{\text{avg}}}.
\end{align*}

The following lemma holds regarding the expected value of the proposed estimator.
\begin{lemma}\label{lemma:13}
The estimator $\hat{\rho}'(l)$ converges to
\begin{align*}
\tilde{\rho}(l) = \frac{\sum_{v_i \in V^*} 1_{{\{l \in \mathcal{L}(i)\}}} d_i^* / d_i}{\sum_{v_i \in V^*} d_i^*/d_i}
\end{align*}
after many steps of our random walk.
\end{lemma}

\begin{proof}
We calculate the expected value of $\Phi_{\text{label}}'$ with respect to $\bm{\pi}$ as follows:
\begin{align*}
\mathbb{E}_{\bm{\pi}}[\Phi_{\text{label}}'] 
= \mathbb{E}_{\bm{\pi}}\left[\frac{1_{{\{l \in \mathcal{L}(x_k)\}}}}{d_{x_k}}\right] 
= \frac{1}{D^*} \sum_{v_i \in V^*} \frac{d_i^*}{d_i} 1_{{\{l \in \mathcal{L}(i)\}}}.
\end{align*}
Since $\Phi'_{\text{avg}}$ converges to the expected value $\mathbb{E}_{\bm{\pi}}[\Phi'_{\text{avg}}] = (\sum_{v_i \in V^*} d_i^*/d_i)/D^*$ because of Theorem \ref{SLLN}, $\hat{d}_{\text{avg}}$ converges to $\tilde{d}_{\text{avg}}$ after many steps of our random walk.
\end{proof}

We have the following proposition.
\begin{proposition}\label{proposition:3}
When the original network $G$ involves no private nodes, two estimators $\hat{\rho}(l)$ and $\hat{\rho}'(l)$ are equal, and two expected values, $\rho^*(l)$ and $\tilde{\rho}(l)$, are equal to the true quantity $\rho(l)$.
\end{proposition}

Finally, we have the following theorem.

\begin{theorem}\label{theorem:5}
If all the public nodes belong to the largest public cluster, we have
\begin{align*}
  \mathbb{E}_{\text{pri}}[\tilde{\rho}(l)] \approx 
  \frac{\mathbb{E}_{\text{pri}}[\sum_{v_i \in V^*} 1_{{\{l \in \mathcal{L}(i)\}}} d_i^* / d_i]}{\mathbb{E}_{\text{pri}}[\sum_{v_i \in V^*} d_i^*/d_i]} 
  = \rho(l).
\end{align*}
\end{theorem}

\begin{proof}
We define a random variable $\tilde{X}_{\text{label}} (i) = 1_{{\{v_i \in V^*\}}} d_i^*/d_i$ for each node $v_i \in V$. 
Let $\tilde{X}_{\text{label}} = \sum_{v_i \in V^*} 1_{{\{l \in \mathcal{L}(i)\}}} d_i^* / d_i$. 
It holds that $\tilde{X}_{\text{label}} = \sum_{v_i \in V} \tilde{X}_{\text{label}} (i)$. 
We obtain the expected value of $\tilde{X}_{\text{label}}$ with respect to $\mathcal{L}_{\text{pri}}$:
\begin{align*}
  \mathbb{E}_{\text{pri}}[\tilde{X}_{\text{label}}] 
  = \sum_{v_i \in V} \text{Pr}[v_i \in V^*]\mathbb{E}_{\text{pri}}\left[\frac{d_i^*}{d_i} 1_{{\{l \in \mathcal{L}(i)\}}}\right]
  = (1-p)^2 \sum_{v_i \in V} 1_{{\{l \in \mathcal{L}(i)\}}}.
\end{align*}
Theorem \ref{theorem:5} holds because of the above equation and equation $\mathbb{E}_{\text{pri}}[\sum_{v_i \in V^*} d_i^*/d_i] = (1-p)^2 n$ (see the proof of Theorem \ref{theorem:4}).
\end{proof}

\subsection{Density of the private label of the node} \label{section:5.5}
The proposed estimator $\hat{\rho}'(l)$ is not applicable to the estimation of the density of the private label of the node (i.e., the fraction of private nodes).
This is because the sample sequence does not contain private nodes.
However, it is possible to intuitively estimate the probability $p$ using the existing and proposed estimators of the network size and average degree each.
Note that the probability $p$ is almost equal to the density of the private label in a large-scale social network under Assumption \ref{assumption:2}.
We denote by $\hat{p}_{\text{size}}$ the estimator of the probability $p$ obtained by the existing and proposed estimators of the network size.
We denote by $\hat{p}_{\text{avg}}$ the estimator of the probability $p$ obtained by the existing and proposed estimators of the average degree.
Based on Lemmas \ref{lemma:5} and \ref{lemma:9} and Theorems \ref{theorem:3} and \ref{theorem:4}, we define estimators $\hat{p}_{\text{size}}$ and $\hat{p}_{\text{avg}}$ as
\begin{align}
\hat{p}_{\text{size}} &\triangleq 1 - \hat{n}/\hat{n}', \label{eq:4} \\
\hat{p}_{\text{avg}} &\triangleq 1 - \hat{d}_{\text{avg}}/{\hat{d}_{\text{avg}}}'. \label{eq:5}
\end{align}

\subsection{Estimation in the Hidden Privacy Model}
In the hidden privacy model, we calculate each estimator using the estimated public degree of each sampled node, $\hat{d}_{x_k}^*$. 
Even in this model, Lemmas \ref{lemma:4}, \ref{lemma:6}, \ref {lemma:8}, \ref {lemma:10}, \ref{lemma:11}, and \ref{lemma:13} hold because of Lemma \ref{lemma:2}.

\section{Experiments} \label{section:6}

We numerically evaluate the proposed estimators using social network datasets. 
We aim to answer the following questions:
\begin{enumerate}
  \item Do the proposed estimators reduce the biases induced by private nodes of the existing estimators for the network size, average degree, and density of the node label? (Section \ref{section:6.2.1})
  \item Do the proposed estimators perform acceptably on social network datasets involving private nodes? (Sections \ref{section:6.2.2} and \ref{section:6.2.3})
  \item How does the proposed method for calculating the public degree of each sampled node affect the estimation accuracy and the number of queries in the hidden privacy model? (Section \ref{section:6.2.4})
  \item Is the number of additional queries generated by private nodes during seed selection small? (Section \ref{section:6.2.5})
\end{enumerate}

\subsection{Datasets}

\subsubsection{Description of the datasets} 
We use five social network datasets, i.e., the YouTube, Pokec, Orkut, Facebook, and LiveJournal datasets. 
For these five datasets, we focus on undirected and connected graphs by the following pre-processing: (1) removing the directions of edges if the original graphs are directed and then (2) deleting the nodes that are not contained in the largest connected component of the original graph.
The pre-processing does not affect any of the following experiments because the above pre-processing is performed before setting privacy labels of the nodes and no processing is added to the graph after setting privacy labels of the nodes.

The YouTube, Orkut, Facebook, and LiveJournal datasets do not contain the privacy label data of each node, respectively.
In contrast, the Pokec dataset contains all the graph data of the network involving private nodes and contains empirical privacy labels of all the nodes \cite{takac2012}.
The Pokec network involves 552,525 private nodes (i.e., approximately 33.8\% of all the nodes). 

We additionally use the dataset of the sample sequence of 1,016,275 public Facebook users obtained by Kurant et al.'s random walk in October 2010 \cite{kurant}.
The random walk is equivalent to our random walk in the ideal model because the Facebook graph as of October 2010 involves a fraction of private nodes and corresponds to the ideal model (see Refs.~\cite{kurant} for details).
This dataset contains the ID, the exact public degree, and the exact degree of each sampled public user, which allows us to compare the existing estimators and the proposed estimators.

\begin{figure*}[t]
	\begin{center}\
		\includegraphics[scale=0.22]{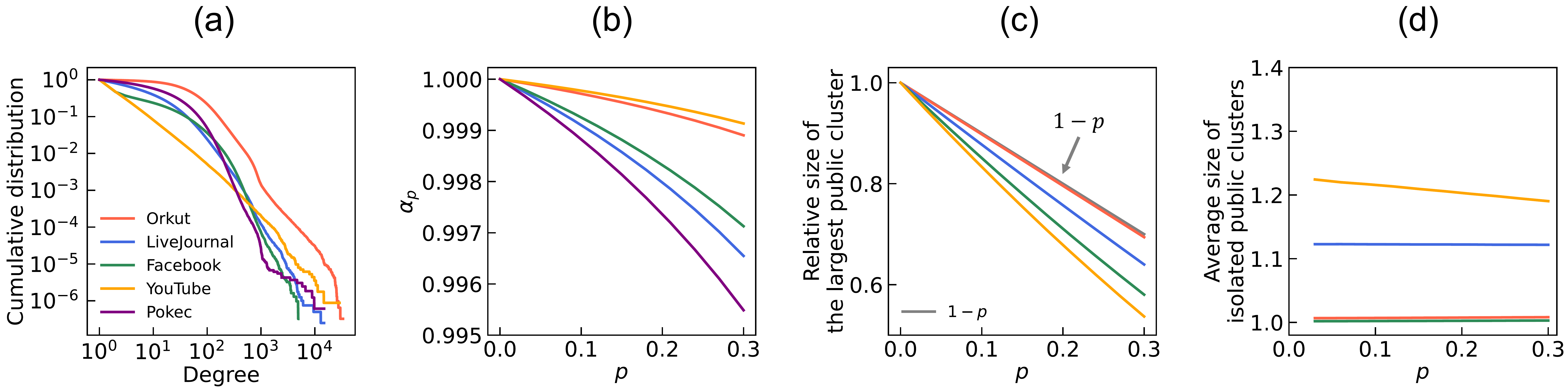}
	\end{center}
	\caption{Four properties of the network datasets. (a) Cumulative degree distribution. (b) Coefficient $\alpha_p$ as a function of $p$. (c) Relative size of the largest public cluster as a function of $p$. (d) Average size of the isolated public clusters as a function of $p$.}
	\label{fig:3}
\end{figure*}

\begin{table}[t]
\caption{Network datasets. $n$: network size, $d_{\text{avg}}$: average degree, $\bar{c}$: clustering coefficient \cite{watts1998}, and $\bar{\ell}$: average shortest-path length between nodes \cite{watts1998}.}
\label{table:1}
\begin{center}
	\begin{tabular}{l | c c c c | c | c}\hline
	Network & $n$ & $d_{\text{avg}}$ & $\bar{c}$ & $\bar{\ell}$ & Privacy-label data & Reference \rule[5pt]{0pt}{5pt} \\ \hline
	YouTube & 1,134,890 & 5.27 & 0.08 & 5.28 & Not contain & \cite{konect} \\
	Pokec & 1,632,803 & 27.32 & 0.11 & 4.68 & Contain & \cite{snap} \\
	Orkut & 3,072,441 & 76.28 & 0.17 & 4.21 & Not contain & \cite{snap} \\
	Facebook & 3,097,165 & 15.28 & 0.10 & 5.16 & Not contain & \cite{nr} \\
	LiveJournal & 3,997,962 & 17.35 & 0.28 & 5.57 & Not contain & \cite{snap} \\ \hline
  	\end{tabular}
\end{center}
\end{table}

\subsubsection{Properties of the network datasets} \label{section:6.1.2}

Table \ref{table:1} shows the network size and average degree for the five network datasets.
Table \ref{table:1} also indicates that each network has a small-world property (i.e., large clustering coefficient and small average shortest-path length between nodes \cite{watts1998}).

Figure \ref{fig:3} shows four additional network properties of each of the network datasets.
Figure  \ref{fig:3}(a) shows the cumulative degree distributions of the five datasets. 
We observe that all the five networks have heavy-tailed degree distributions.
Figure \ref{fig:3}(b) shows the coefficient $\alpha_p$, defined in Eq. \eqref{eq:3}, as a function of $p$ for the five datasets. 
We find that $\alpha_p$ is almost equal to 1.0 for $0.0 \leq p \leq 0.3$.
Figure \ref{fig:3}(c) shows the relative size of the largest public cluster averaged over 1,000 independent and random sets of private nodes as a function of $p$ for the four datasets.
The gray solid line represents the expected upper limit, $1-p$.
We observe that a large fraction of public nodes belong to $C^*$ for $0.0 \leq p \leq 0.3$ for the four datasets, which is qualitatively consistent with the previous study \cite{albert2000}.
Figure \ref{fig:3}(d) shows the average size of the isolated public clusters (i.e., the public clusters other than the largest public cluster) averaged over 1,000 independent and random sets of private nodes as a function of $p$ for the four datasets. 
The average size of the isolated public clusters is considerably small for the four datasets, which is qualitatively consistent with the finding in the previous study \cite{albert2000}.
We also observe that all the public nodes belong to the largest public cluster of the Pokec network (i.e., the network has no isolated public clusters), where the empirical privacy label is given to each node.

\begin{figure*}[]
	\begin{center}
		\includegraphics[scale=0.15]{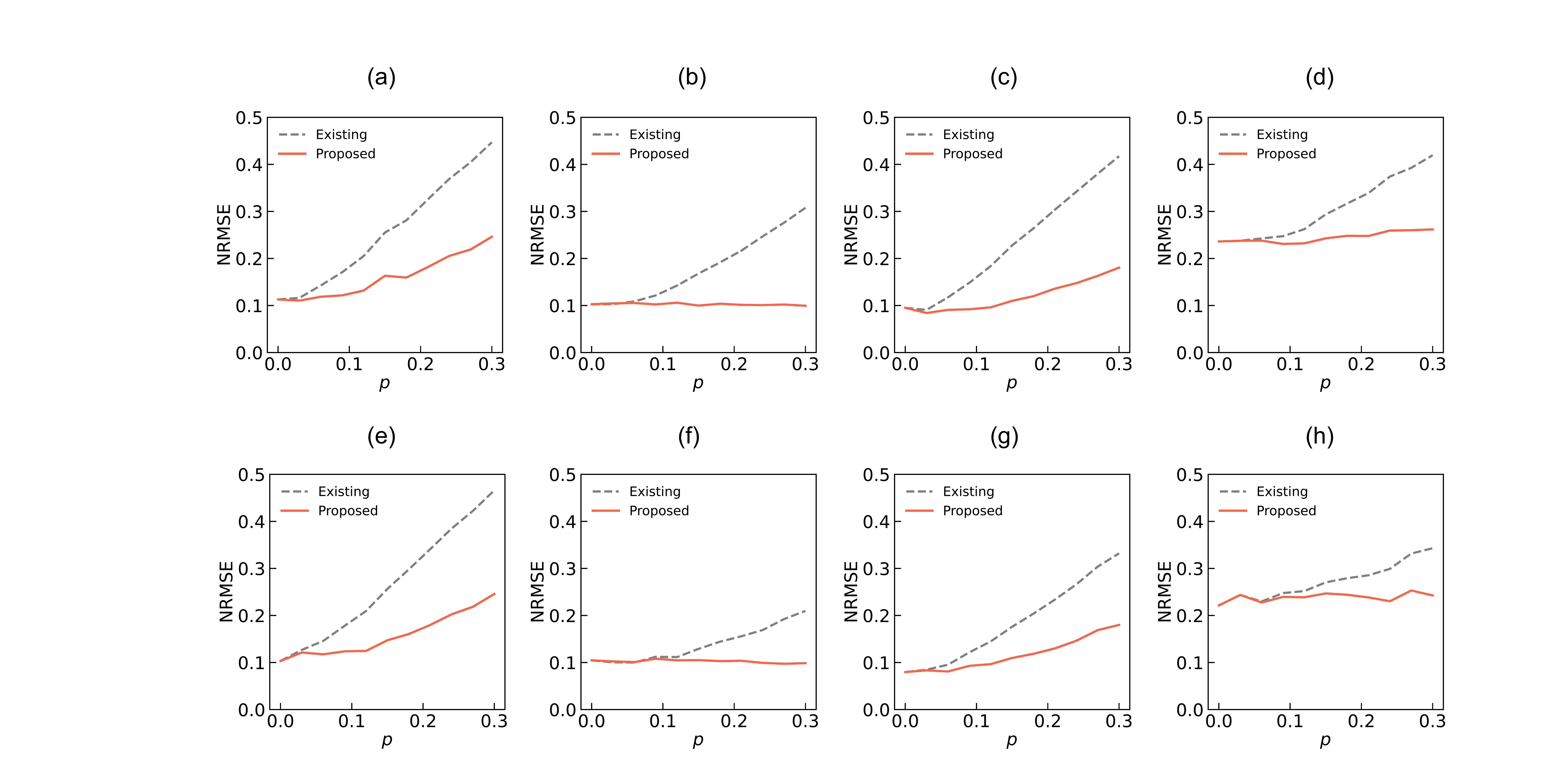}
	\end{center}
	\caption{NRMSEs of the existing and proposed estimators for the network size as a function of $p$. Panels (a) and (e) show the results for the YouTube dataset; panels (b) and (f) show the results for the Orkut dataset; panels (c) and (g) show the results for the Facebook dataset; panels (d) and (h) show the results for the LiveJournal dataset. Panels (a)--(d) show the results in the ideal model, and panels (e)--(h) show the results in the hidden privacy model. We set the sample size, $r$, as $1\%$ of the number of nodes.}
      \label{fig:4}
      \vspace{3mm}
      \begin{center}
		\includegraphics[scale=0.15]{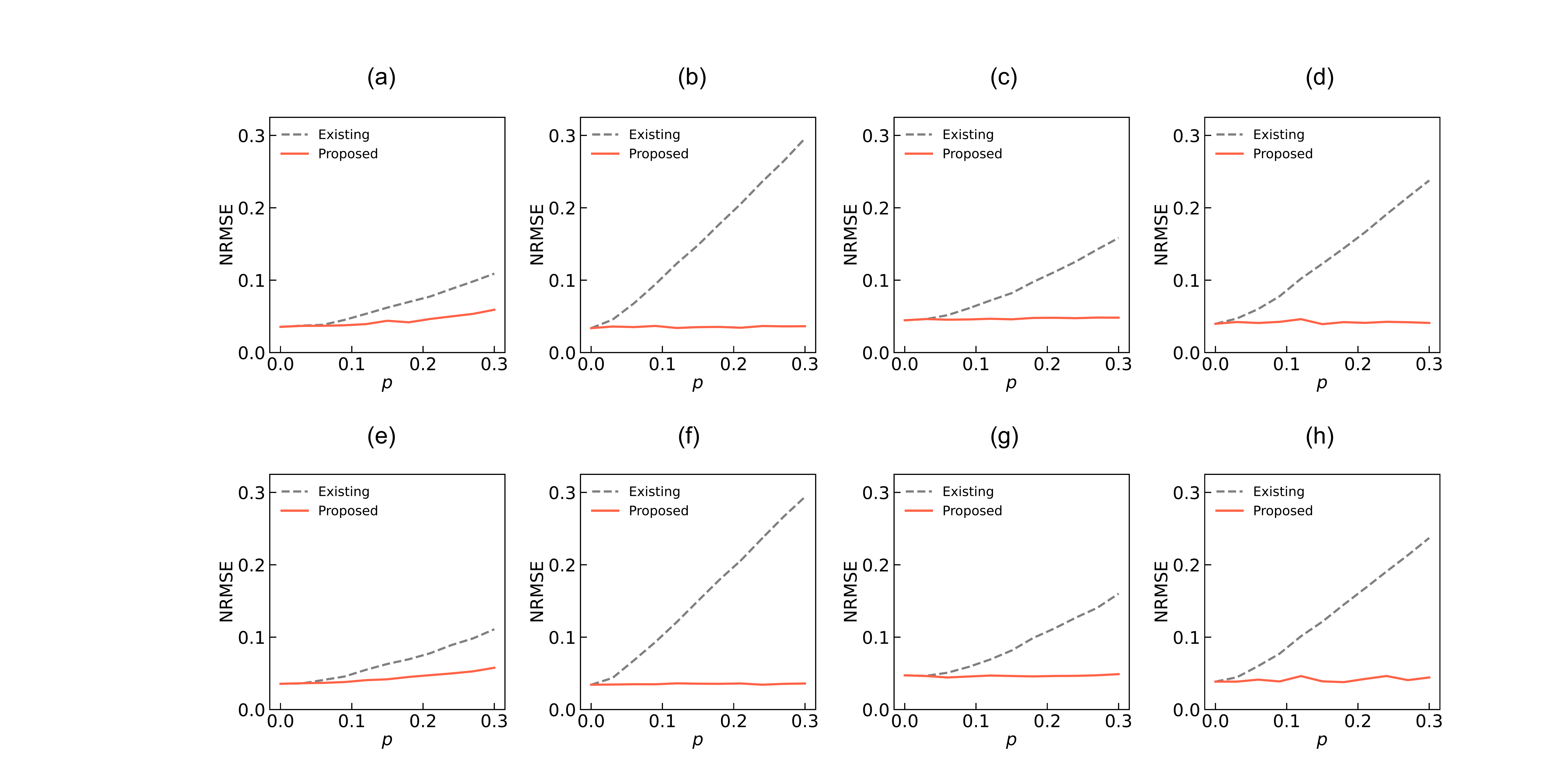}
	\end{center}
      \caption{NRMSEs of the existing and proposed estimators for the average degree as a function of $p$. Panels (a) and (e) show the results for the YouTube dataset; panels (b) and (f) show the results for the Orkut dataset; panels (c) and (g) show the results for the Facebook dataset; panels (d) and (h) show the results for the LiveJournal dataset. Panels (a)--(d) show the results in the ideal model, and panels (e)--(h) show the results in the hidden privacy model. We set the sample size, $r$, as $1\%$ of the number of nodes.}
      \label{fig:5}
\end{figure*}

\begin{figure*}[t]
	\begin{center}
		\includegraphics[scale=0.16]{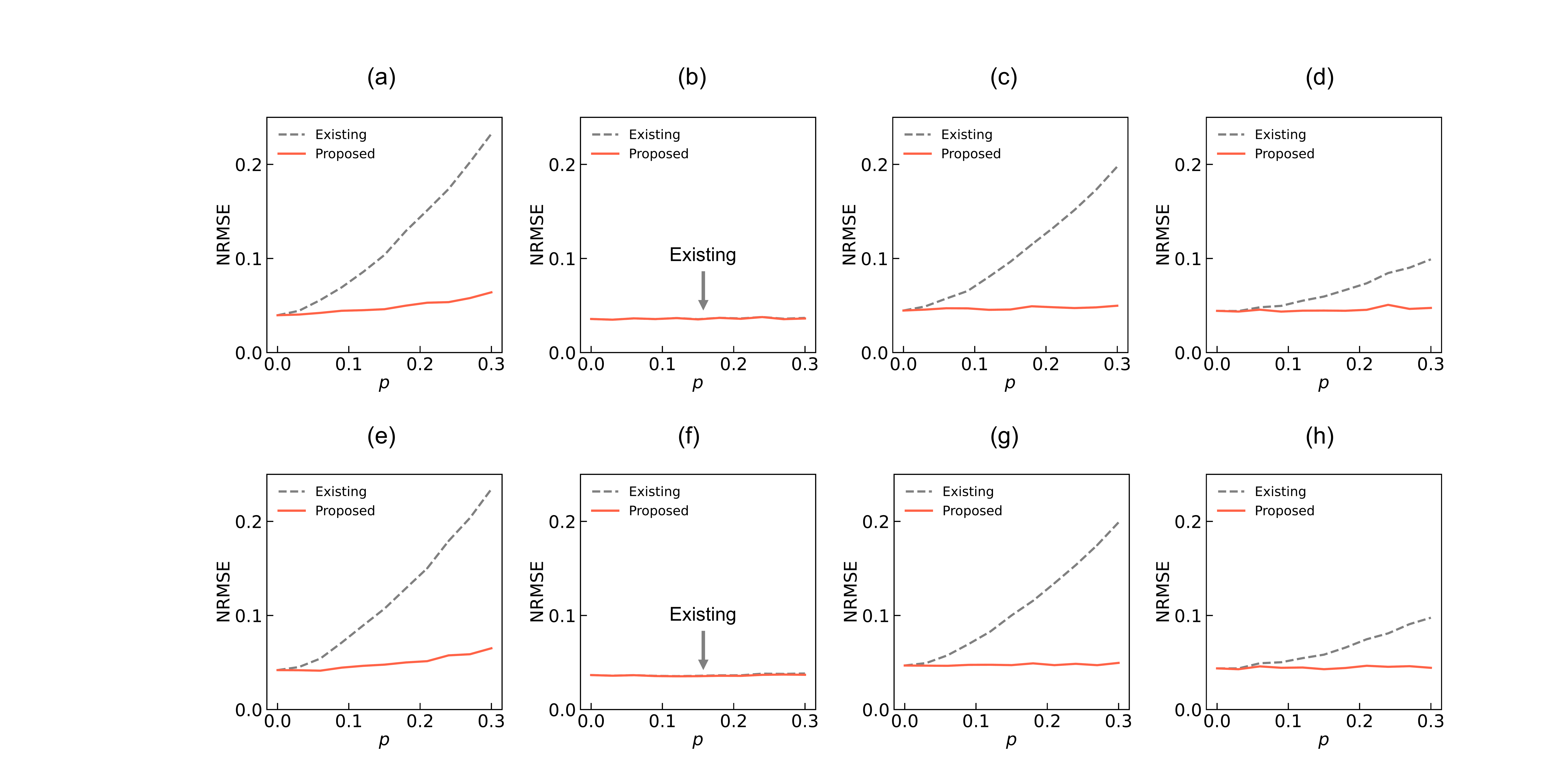}
	\end{center}
	\caption{NRMSEs of the existing and proposed estimators for the cumulative degree distribution as a function of $p$. Panels (a) and (e) show the results for the YouTube dataset; panels (b) and (f) show the results for the Orkut dataset; panels (c) and (g) show the results for the Facebook dataset; panels (d) and (h) show the results for the LiveJournal dataset. Panels (a)--(d) show the results in the ideal model, and panels (e)--(h) show the results in the hidden privacy model. We set the sample size, $r$, as $1\%$ of the number of nodes. We indicate the curves by an arrow and label when two curves heavily overlap each other.}
      \label{fig:6}
\end{figure*}

\begin{figure*}[t]
	\begin{center}
		\includegraphics[scale=0.15]{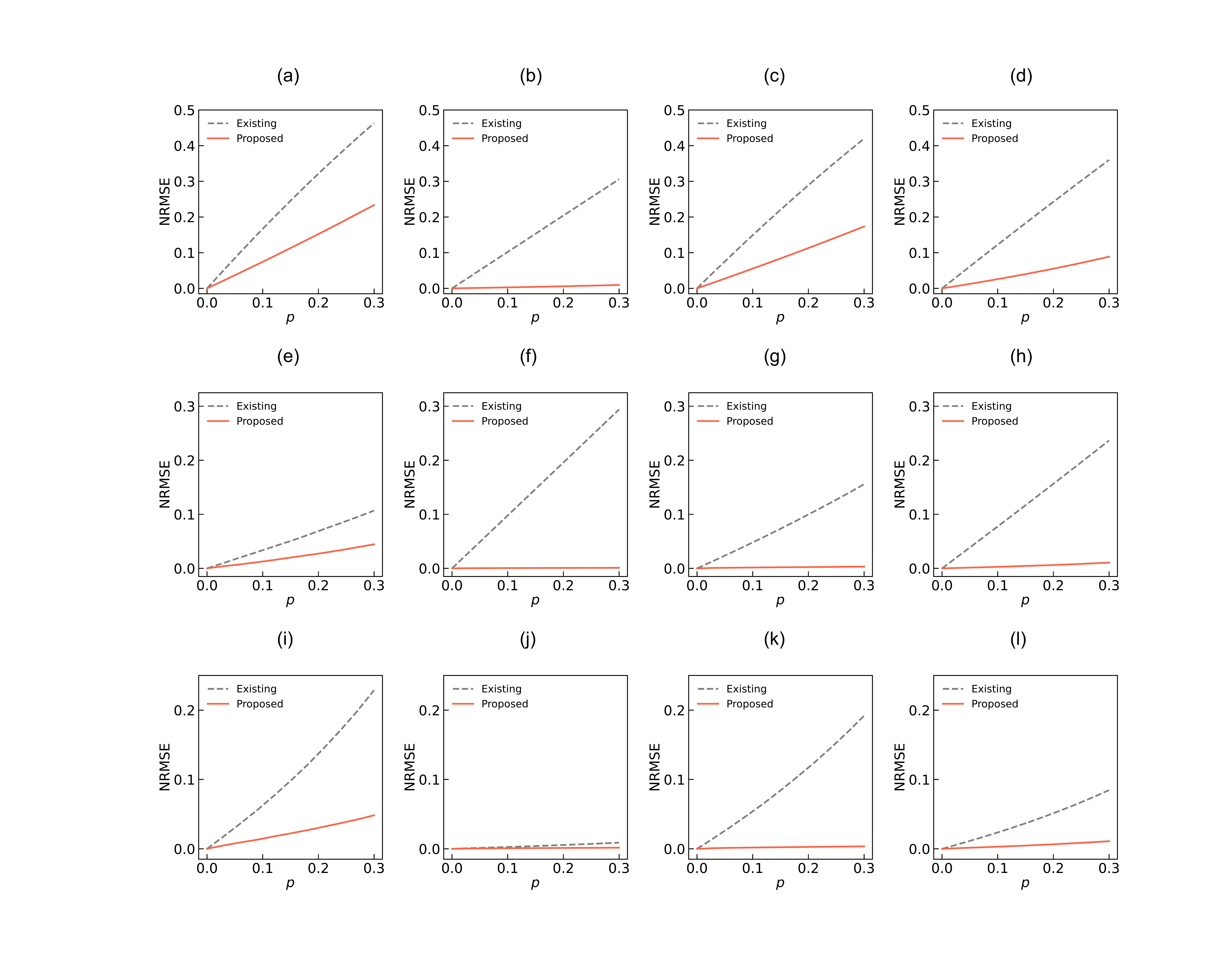}
	\end{center}
	\caption{NRMSEs of the expected values of the existing and proposed estimators as a function of $p$.
Panels (a), (e), and (i) show the results for the YouTube dataset; panels (b), (f), and (j) show the results for the Orkut dataset; panels (c), (g), and (k) show the results for the Facebook dataset; panels (d), (h), and (l) show the results for the LiveJournal dataset. 
Panels (a)--(d) show the results for the network size; panels (e)--(h) show the results for the average degree; panels (i)--(l) show the results for the cumulative degree distribution.}
      \label{fig:7}
\end{figure*}

\subsection{Results}

\subsubsection{Performance of the proposed estimators when varying the fraction of private nodes} \label{section:6.2.1}
We compare the performance of the existing and proposed estimators for the network size, average degree, and density of the node label using YouTube, Orkut, Facebook, and LiveJournal datasets.
We set the indicator function used in the estimators of the density of the node label to ensure that it returns 1 if a node has degree $k$ or larger and 0 otherwise, which is equivalent to estimating the cumulative degree distribution, $\{P(d)\}_d$ \cite{gjoka2011, gjoka2010, ribeiro2010, lee2012}.

It is practically important for an estimator to have a low bias in a single run and a small variance across multiple runs.
The normalized mean squared error (NRMSE) evaluates both the bias and variance of an estimator \cite{lee2012, katzir2013, wang2014, chen2016}.
For the network size and average degree, we measure the NRMSE of an estimator given by $\sqrt{\mathbb{E}[(\hat{x}/x - 1)^2]}$, where $x$ denotes the exact value and $\hat{x}$ denotes the estimator.
For the cumulative degree distribution, we measure the NRMSE of an estimator given by $\sqrt{\mathbb{E}[(D_{P(d)})^2]}$, where $D_{P(d)} = \sum_{d} |\hat{P}(d) - P(d)|$ is the normalized $L^1$ distance between the exact distribution $\{P(d)\}_{d}$ and the estimated distribution $\{\hat{P}(d)\}_{d}$.

We perform a single run on the YouTube, Orkut, Facebook, and LiveJournal datasets as follows.
First, we independently and randomly set the privacy label of each node as private with a given probability $p$ and otherwise public, according to Assumption \ref{assumption:2}.
Second, we choose a seed uniformly at random from the nodes in the largest public cluster.
Third, we perform our random walk with a length $r$ of $1\%$ of the number of nodes.
Finally, we calculate the existing and proposed estimators from the sampling list.
For the given $p$, we estimate the NRMSE of each estimator over 1,000 independent runs.
We vary the probability $p$ from 0.0 to 0.30 in increments of 0.03 because there were actually tens of percentages of private nodes in empirical social networks \cite{catanese2011, takac2012, dey2012, buccafurri2015}.

Figure \ref{fig:4} shows the NRMSEs of the existing and proposed estimators for the network size as a function of $p$.
The following observations apply to both access models.
First, the NRMSEs of both estimators are exactly the same when $p = 0$, which is consistent with Proposition \ref{proposition:1}. 
Second, more importantly, the proposed estimator typically improves the NRMSE when $p > 0$.
For example, the proposed estimator improves the NRMSE by approximately 67.7\% (i.e., from 0.308 to 0.100) when $p = 0.3$ in Fig. \ref{fig:4}(b).
These observations qualitatively remain the same for the estimators of the average degree (see Fig. \ref{fig:5}) and the cumulative degree distribution (see Fig. \ref{fig:6}), except for Figs. \ref{fig:6}(b) and \ref{fig:6}(f).

To further investigate the performance of the proposed estimators, we observe the NRMSEs of expected values of the existing and proposed estimators for each property (see Lemmas \ref{lemma:4}, \ref{lemma:6}, \ref{lemma:8}, \ref{lemma:10}, \ref{lemma:11}, and \ref{lemma:13} for the expected value of each estimator).
For the given $p$, we calculate the NRMSEs of the expected values of the existing and proposed estimators over 1,000 random sets of privacy labels of nodes.
Note that the expected value of each estimator does not depend on the access model.

Figure \ref{fig:7}(a)--(d) shows the NRMSEs of expected values of the existing and proposed estimators for the network size as a function of probability $p$.
The following observations apply to the four datasets.
First, the NRMSEs of both estimators are equal to zero when $p = 0$, which is consistent with Proposition \ref{proposition:1}.
Second, the proposed estimator improves the NRMSE of the expected value when $p > 0$.
These observations are qualitatively the same for the estimators of the average degree (see Fig. \ref{fig:7}(e)--(h)) and the estimators of the cumulative degree distribution (see Fig. \ref{fig:7}(i)--(l)).
The proposed estimator for the cumulative degree distribution slightly reduces the bias induced by private nodes for the Orkut dataset (see Fig. \ref{fig:7}(j)), which is qualitatively consistent with the little improvement in the NRMSE of the proposed estimator in Figs. \ref{fig:6}(b) and (f).

For the Orkut dataset, where almost all the public nodes belong to the largest public cluster (see Fig. \ref{fig:3}(c)), the proposed estimators for the network size, average degree, and cumulative degree distribution have approximately no bias for any $0.0 \leq p \leq 0.3$ (see Figs. \ref{fig:7}(b), \ref{fig:7}(f), and \ref{fig:7}(j)).
These results support Theorems \ref{theorem:3}, \ref{theorem:4}, and \ref{theorem:5}.
The biases of the existing estimators for the network size and average degree approximately increase linearly with $p$ (see Fig.~\ref{fig:7}(b) and Fig.~\ref{fig:7}(f)), which supports Lemmas \ref{lemma:5} and \ref{lemma:9}.
The existing estimator for the cumulative degree distribution has approximately no bias for any $0.0 \leq p \leq 0.3$ (see Fig.~\ref{fig:7}(j)), which supports Lemma \ref{lemma:12}.

The existing and proposed estimators have biases induced by the isolated public clusters because we assume that a random walk is allowed to traverse only the largest public cluster.
In fact, for the YouTube dataset, where the relative size of the largest public cluster is the smallest among the four datasets (see Fig. \ref{fig:3}(c)), the NRMSEs of the existing and proposed estimators relatively increase as $p$ increases (see Figs. \ref{fig:7}(a), \ref{fig:7}(e), and \ref{fig:7}(i)).
Nevertheless, the proposed estimator still has a smaller NRMSE than the existing estimator for $p > 0$.

\subsubsection{Estimation on the Pokec dataset} \label{section:6.2.2}
We evaluate the proposed estimators using the Pokec network dataset involving private nodes \cite{snap, takac2012}.
We perform a single run on the Pokec dataset as follows.
First, we apply the privacy label contained in the dataset to each node.
Second, we choose a seed uniformly at random from the nodes in the largest public cluster of the network.
Fourth, we perform our random walk with length $r$.
Finally, we calculate the existing and proposed estimators from the sampling list.
We vary the length $r$ from $0.5\%$ of the number of nodes to $5\%$ of the number of nodes in increments of $0.5\%$ of the number of nodes. 
For the given $r$, we estimate the NRMSE of each estimator over 1,000 independent runs.

Figure \ref{fig:8} shows the NRMSE of the existing and proposed estimators for the three properties as a function of the sample size.
For each property, the proposed estimator improves the NRMSE for any sample size in both the access models. 
For example, the proposed estimator for the network size improves the NRMSE by approximately $92.6\%$ (i.e., from 0.339 to 0.025) in the case of $5\%$ sample size in the ideal model (see Fig. \ref{fig:8}(a)). 

Table \ref{table:2} shows the errors (i.e., the relative error or the $L^1$ distance) of the expected values of the existing and proposed estimators for each property.
The proposed estimators improve the error by 97.3\% for the network size, 87.5\% for the average degree, and 32.1\% for the cumulative degree distribution.
Furthermore, for each property, the proposed estimator has the expected value that is almost equal to the true quantity of the whole Pokec network involving private nodes.
Although Assumption \ref{assumption:2} does not hold for the Pokec network, the proposed estimators yield the results which are consistent with Theorems \ref{theorem:3}, \ref{theorem:4}, and \ref{theorem:5}.

\begin{figure*}[t]
	\begin{center}
		\includegraphics[scale=0.16]{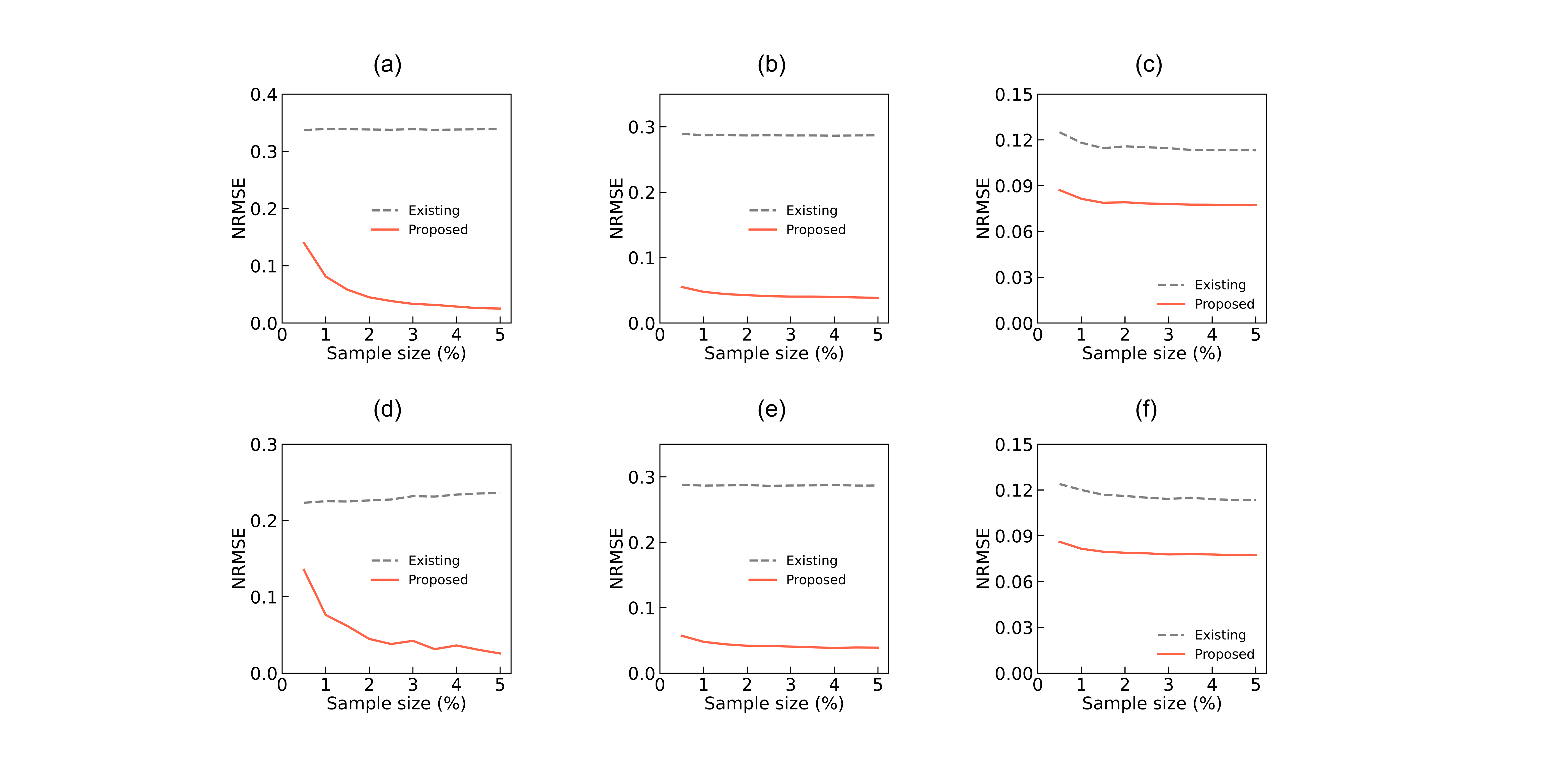}
	\end{center}
      \caption{NRMSEs of the existing and proposed estimators as a function of the sample size on the Pokec dataset. Panels (a) and (d) show the results for the network size; panels (b) and (e) show the results for the average degree; panels (c) and (f) show the results for the cumulative degree distribution. Panels (a)--(c) show the results in the ideal model, and panels (d)--(f) show the results in the hidden privacy model.}
      \label{fig:8}
\end{figure*}

\begin{table}[t]
\begin{center}
\caption{Expected errors of the existing and proposed estimators for the three properties on the Pokec dataset. For the network size and average degree, the error shows the relative error. For the cumulative degree distribution, the error shows the $L^1$ distance.}
\label{table:2}
\begin{tabular}{l | c c} \hline
	Property & Existing & Proposed \\ \hline
	Network size & 0.338 & 0.009 \\ 
	Average degree & 0.287 & 0.036 \\ 
	Cumulative degree distribution & 0.112 & 0.076 \\ \hline
  \end{tabular}
\end{center}
\end{table}

\subsubsection{Estimation on the Facebook sample dataset} \label{section:6.2.3}

We use the sample sequence of 1,016,275 public nodes obtained by Kurant et al.'s random walk on the Facebook graph in October 2010 \cite{kurant}.
The dataset contains the ID, the exact public degree, and the exact degree of each sampled public node.
Therefore, we compare the existing and proposed estimators for the network size, average degree, and cumulative degree distribution of the Facebook graph as of October 2010.

Table \ref{table:3} shows the existing and proposed estimators for the network size and average degree. 
It is difficult to calculate the error of each estimator because the true quantities of the Facebook graph as of October 2010 are unknown. 
However, we consider that the estimates are reasonable considering two findings in almost the same period.
First, Facebook reported that there were 500 million active users as of July 2010 \cite{facebook_500m}, which indicates that there were at least 500 million users, including inactive users, at that time. 
Notably, the estimates of the network size shown in Table \ref{table:3} count both active and inactive users. 
Our estimate, i.e., 657 million, is greater than 500 million, and we speculate that the difference (i.e., approximately 157 million) mainly comprises inactive users.
Second, Catanese et al. obtained an unbiased estimate of the proportion of private nodes as 0.266 using a uniform sample of Facebook users in August 2010 \cite{catanese2011}\footnote{When Catanese et al. performed a uniform sampling of users on Facebook during August 2010, the user id was 32 bit. As mentioned in Refs. \cite{gjoka2011, gjoka2010}, shortly after that, Facebook's user id went to 64 bit, and uniform sampling in the 64-bit space is typically infeasible.}. 
Based on Table \ref{table:3}, we obtain the estimates of the proportion of private nodes, i.e., $\hat{p}_{\text{size}}$ defined in Eq. \eqref{eq:4} and $\hat{p}_{\text{avg}}$ defined in Eq. \eqref{eq:5}, as 0.269 and 0.255, respectively.
These two estimates are considerably close to the ground truth value of 0.266.
Finally, Fig. \ref{fig:9} shows the existing and proposed estimators for the cumulative degree distribution.
We observe that two estimates heavily overlap each other.
This result is qualitatively consistent with our theoretical results of Lemma \ref{lemma:11} and Theorem \ref{theorem:5}.

\begin{table}[t]
\begin{center}
\caption{Estimates of the network size and average degree obtained from the Facebook sample dataset.}
\label{table:3}
\begin{tabular}{l | c c} \hline
	Property & Existing & Proposed \\ \hline
	Network size & 480,298,540 & 656,874,081 \\ 
	Average degree & 102.07 & 137.03 \\ \hline
  \end{tabular}
\end{center}
\end{table}

\begin{figure}
	\begin{center}
		\includegraphics[scale=0.2]{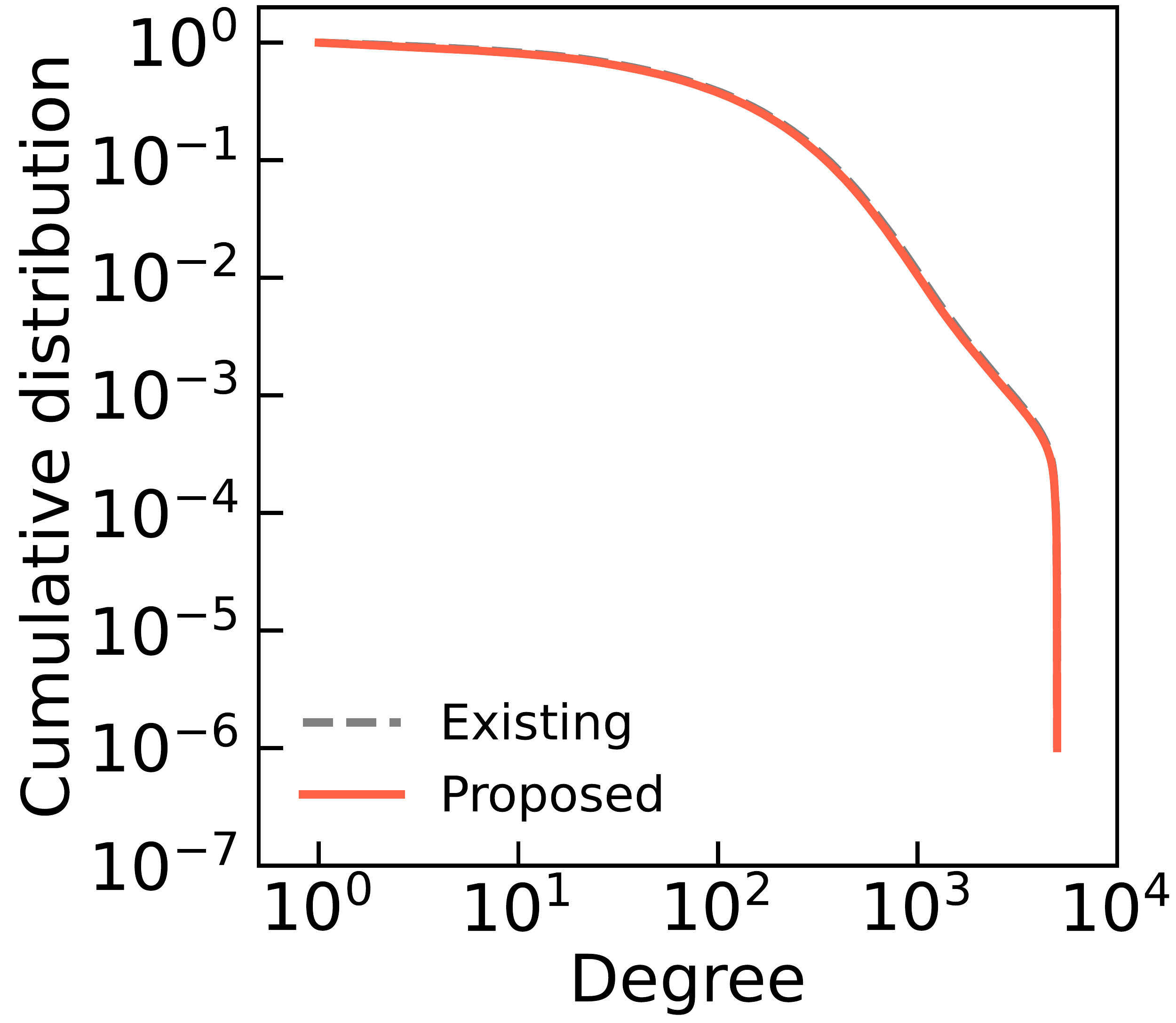}
	\end{center}
	\caption{Estimates of the cumulative degree distribution obtained from the Facebook sample dataset. Two curves heavily overlap each other.}
\label{fig:9}
\end{figure}

\subsubsection{Effectiveness of the proposed method for calculating public degree} \label{section:6.2.4}

\begin{figure*}[t]
	\begin{center}
		\includegraphics[scale=0.15]{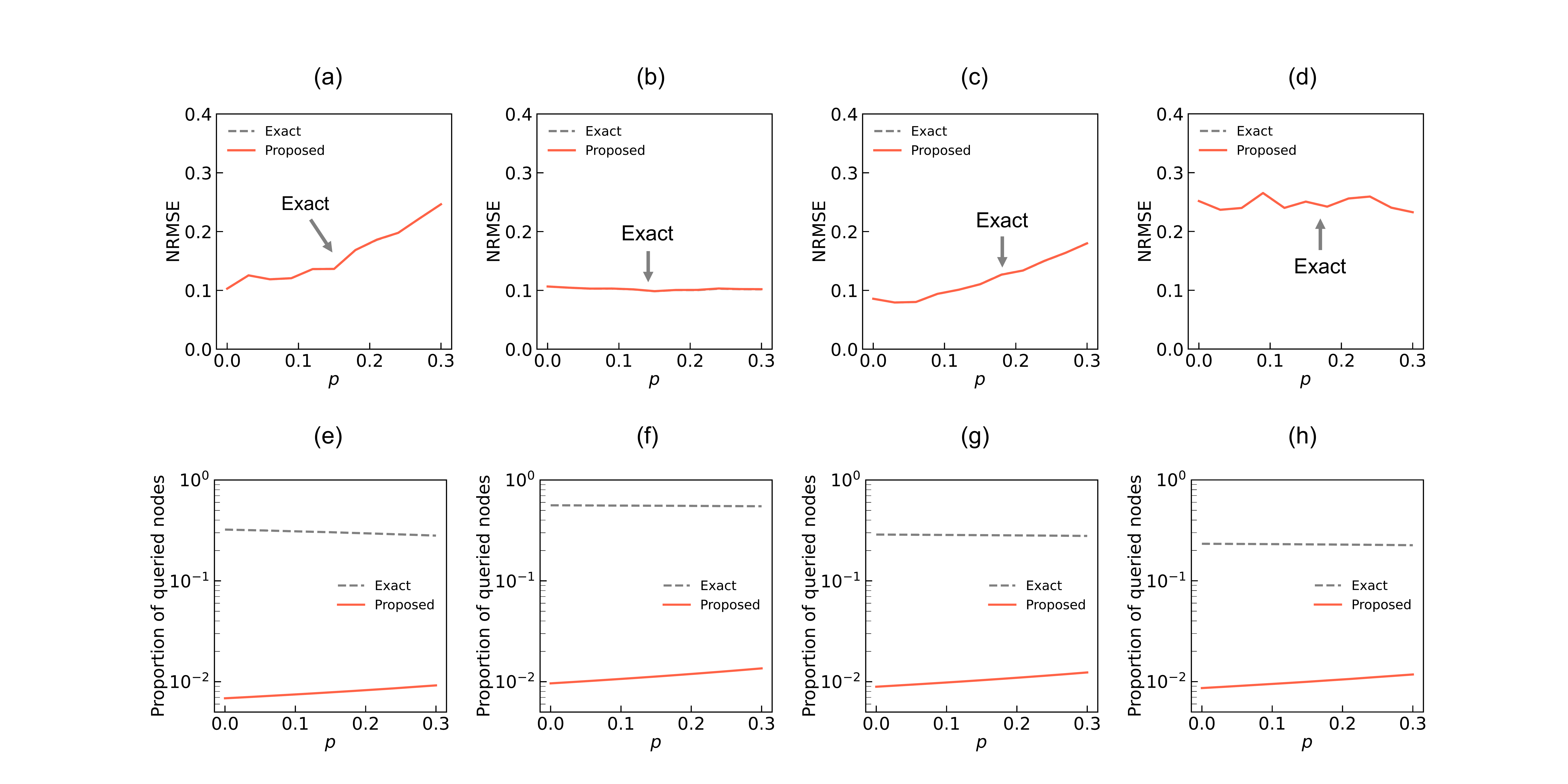}
	\end{center}
      \caption{Effects of the proposed method for calculating the public degree of each sampled node in the hidden privacy model. Panels (a) and (d) show the results for the YouTube dataset; panels (b) and (e) show the results for the Orkut dataset; panels (c) and (f) show the results for the Facebook dataset; panels (d) and (h) show the results for the LiveJournal dataset. Panels (a)--(d) show the NRMSEs of the estimators of the network size using the exact and proposed methods as a function of $p$. Panels (e)--(h) show the proportion of queried nodes using the exact and proposed methods as a function of $p$. We set the sample size as $1\%$ of the number of nodes. We indicate the curves by an arrow and label when two curves heavily overlap each other.}
      \label{fig:10}
\end{figure*}

We evaluate the proposed method for calculating the public degree of each sampled node in the hidden privacy model. 
The proposed estimator for the network size requires the public degree of each sampled node (see Section \ref{section:5.2.2}). 
Therefore, we compare the NRMSE and the proportion of queried nodes for the proposed size estimator using the proposed method and exact method, respectively.
In the exact method, one queries all the neighbors of each sampled node to calculate the exact public degree.
We perform the simulations on the YouTube, Orkut, Facebook, and LiveJournal datasets using the same procedure followed in Section \ref{section:6.2.1}.

Figure \ref{fig:10} shows the results for the four datasets.
The following observations apply to all the four datasets.
First, the proposed method achieves almost the same NRMSE as the exact method (see Fig. \ref{fig:10}(a)--(d)).
Second, although the exact method queries tens of percent of nodes which are much greater than the $1\%$ sample size, the proposed method queries approximately $1\%$ nodes (see Fig. \ref{fig:10}(e)--(h)).
These results support Lemma \ref{lemma:3}.
Therefore, the proposed method reduces the proportion of queried nodes by tens of percent while maintaining almost the estimation accuracy compared with the exact method.

\subsubsection{Seed selection} \label{section:6.2.5}
Thus far, we have assumed that we have access to some arbitrary node that belongs to the largest public cluster to begin our random walk (see Assumption \ref{assumption:3}).
In practical scenarios, we may require additional queries in seed selection by restarting a random walk from another seed in the following two cases.
The first case is when a seed is a private node.
This is the case, for example, when one selects nodes $v_3$ or $v_7$ as a seed in the graph shown in Fig. \ref{fig:1}.
The second case is when a seed is on an isolated public cluster.
This is the case, for example, when one selects nodes $v_6$, $v_8$, or $v_9$ as a seed in the graph shown in Fig. \ref{fig:1}.

We consider that the number of queries generated in each case is sufficiently small. 
We generate a small number of queries in the first case because (i) the proportion of private nodes is generally smaller than that of public nodes in empirical social networks (e.g., 27\% on Facebook \cite{catanese2011} and 34\% on Pokec \cite{takac2012}) and (ii) one query is enough to check the privacy label of a given seed.
In the second case, we may generate a small number of queries under Assumption \ref{assumption:2} owing to the following two natures of empirical networks having heavy-tailed degree distributions \cite{albert2000}. 
First, most public nodes belong to the largest public cluster. 
In our simulations, even if $p = 0.3$, 99.1\% on Orkut, 91.4\% on LiveJournal, 82.8\% on Facebook, and 76.7\% on YouTube of the public nodes belong to the largest public cluster of the original network, respectively (see Fig. \ref{fig:3}(c)). 
Second, the average size of the isolated public clusters is considerably smaller than the size of the largest public cluster.
In our simulations, the average size is only approximately one for every probability $p$ on the four datasets (see Fig. \ref{fig:3}(d)). 

The two empirical datasets involving private nodes may support Assumption \ref{assumption:3}.
Because all the public nodes belong to the largest public cluster of the Pokec network, the second case does not occur on that network.
The Facebook sample dataset yields an estimate of 657 million users and contains one million unique public users.

\section{Conclusions}
In this study, we proposed a framework based on a random walk for estimating properties of social networks involving private nodes.
Most existing estimators assume a network only composed of public nodes, which induces biases due to private nodes in networks involving private nodes.
We extended a simple random walk and the existing estimators for three properties (i.e., the network size, average degree, and density of the node label) to the case of social networks involving private nodes based on the three assumptions with respect to private nodes.
We found that the proposed estimators provide reasonable estimates on the two social network datasets involving private nodes. 

A limitation of the present study is the assumption that each node independently becomes a private node with a given probability.
Private nodes may be biased toward nodes with high or low degrees in empirical social networks.
If private nodes are biased toward high-degree nodes, the entire network may contain a large number of extremely small public clusters (see Ref. \cite{albert2000}); in this case, crawling-based sampling may hardly work on the network.
However, given that previous studies were able to sample more than tens of thousands of users in an OSN \cite{ahn2007, mislove2007, wilson2009, gjoka2010, gjoka2011}, this case may not often happen in practical scenarios.
On the other hand, if private nodes are biased toward low-degree nodes, the results for the proposed estimators may not be substantially different from those presented in this study.
This is because empirical OSNs often have heavy-tailed degree distributions (e.g., see Fig. \ref{fig:3}(a)), which will contribute to the increase in the number of private nodes with low degrees when each node becomes a private node independently at random.
Further investigating the distribution of private nodes on empirical social networks and the effects of private nodes biased toward high- or low-degree nodes on the proposed framework warrant future work.

The present framework seems to be effective for empirical social networks with heavy-tailed degree distributions.
In such networks, a large fraction of public nodes belong to the largest public cluster, which may allow us to reasonably estimate properties of the entire network involving private nodes and support the validity of Assumption \ref{assumption:3}.
To further investigate effects of the heavy-tailed degree distribution and other network properties (e.g., small-world property \cite{watts1998}) on the present framework, it may be useful to compare the performance of the proposed estimators among synthetic network models including Erd\H{o}s-R\'{e}nyi model \cite{erdos1959}, Watts-Strogatz model \cite{watts1998}, and Barab\'{a}si-Albert model \cite{barabasi1999}.

\section*{Acknowledgments}
Kazuki Nakajima was supported by JSPS KAKENHI Grant Number JP21J10415.
Kazuyuki Shudo was supported by JSPS KAKENHI Grant Number JP21H04872.

\bibliography{ref.bib}

\begin{thebibliography}{10}
\urlstyle{rm}
\expandafter\ifx\csname url\endcsname\relax
  \def\url#1{\texttt{#1}}\fi
\expandafter\ifx\csname urlprefix\endcsname\relax\def\urlprefix{URL }\fi
\expandafter\ifx\csname doiprefix\endcsname\relax\def\doiprefix{DOI: }\fi
\providecommand{\bibinfo}[2]{#2}
\providecommand{\eprint}[2][]{\url{#2}}

\bibitem{ahn2007}
\bibinfo{author}{Ahn, Y.-Y.}, \bibinfo{author}{Han, S.}, \bibinfo{author}{Kwak,
  H.}, \bibinfo{author}{Moon, S.} \& \bibinfo{author}{Jeong, H.}
\newblock \bibinfo{title}{{Analysis of Topological Characteristics of Huge
  Online Social Networking Services}}.
\newblock In \emph{\bibinfo{booktitle}{Proceedings of the 16th International
  Conference on World Wide Web}}, \bibinfo{pages}{835--844}
  (\bibinfo{year}{2007}).

\bibitem{mislove2007}
\bibinfo{author}{Mislove, A.}, \bibinfo{author}{Marcon, M.},
  \bibinfo{author}{Gummadi, K.~P.}, \bibinfo{author}{Druschel, P.} \&
  \bibinfo{author}{Bhattacharjee, B.}
\newblock \bibinfo{title}{{Measurement and Analysis of Online Social
  Networks}}.
\newblock In \emph{\bibinfo{booktitle}{Proceedings of the 7th ACM SIGCOMM
  Conference on Internet Measurement}}, \bibinfo{pages}{29--42}
  (\bibinfo{year}{2007}).

\bibitem{wilson2009}
\bibinfo{author}{Wilson, C.}, \bibinfo{author}{Boe, B.}, \bibinfo{author}{Sala,
  A.}, \bibinfo{author}{Puttaswamy, K.~P.} \& \bibinfo{author}{Zhao, B.~Y.}
\newblock \bibinfo{title}{User interactions in social networks and their
  implications}.
\newblock In \emph{\bibinfo{booktitle}{Proceedings of the 4th ACM European
  Conference on Computer Systems}}, \bibinfo{pages}{205--218}
  (\bibinfo{year}{2009}).

\bibitem{gjoka2010}
\bibinfo{author}{Gjoka, M.}, \bibinfo{author}{Kurant, M.},
  \bibinfo{author}{Butts, C.~T.} \& \bibinfo{author}{Markopoulou, A.}
\newblock \bibinfo{title}{{Walking in Facebook: A Case Study of Unbiased
  Sampling of OSNs}}.
\newblock In \emph{\bibinfo{booktitle}{2010 Proceedings IEEE INFOCOM}},
  \bibinfo{pages}{1--9} (\bibinfo{year}{2010}).

\bibitem{kwak2010}
\bibinfo{author}{Kwak, H.}, \bibinfo{author}{Lee, C.}, \bibinfo{author}{Park,
  H.} \& \bibinfo{author}{Moon, S.}
\newblock \bibinfo{title}{{What is Twitter, a Social Network or a News Media?}}
\newblock In \emph{\bibinfo{booktitle}{Proceedings of the 19th International
  Conference on World Wide Web}}, \bibinfo{pages}{591--600}
  (\bibinfo{year}{2010}).

\bibitem{gjoka2011}
\bibinfo{author}{Gjoka, M.}, \bibinfo{author}{Kurant, M.},
  \bibinfo{author}{Butts, C.~T.} \& \bibinfo{author}{Markopoulou, A.}
\newblock \bibinfo{journal}{\bibinfo{title}{{Practical Recommendations on
  Crawling Online Social Networks}}}.
\newblock {\emph{\JournalTitle{IEEE Journal on Selected Areas in
  Communications}}} \textbf{\bibinfo{volume}{29}}, \bibinfo{pages}{1872--1892}
  (\bibinfo{year}{2011}).

\bibitem{fukuda2022}
\bibinfo{author}{Fukuda, M.}, \bibinfo{author}{Nakajima, K.} \&
  \bibinfo{author}{Shudo, K.}
\newblock \bibinfo{journal}{\bibinfo{title}{Estimating the bot population on
  twitter via random walk based sampling}}.
\newblock {\emph{\JournalTitle{IEEE Access}}} \textbf{\bibinfo{volume}{10}},
  \bibinfo{pages}{17201--17211} (\bibinfo{year}{2022}).

\bibitem{katzir2011}
\bibinfo{author}{Katzir, L.}, \bibinfo{author}{Liberty, E.} \&
  \bibinfo{author}{Somekh, O.}
\newblock \bibinfo{title}{Estimating sizes of social networks via biased
  sampling}.
\newblock In \emph{\bibinfo{booktitle}{Proceedings of the 20th International
  Conference on World Wide Web}}, \bibinfo{pages}{597--606}
  (\bibinfo{year}{2011}).

\bibitem{katzir2013}
\bibinfo{author}{Hardiman, S.~J.} \& \bibinfo{author}{Katzir, L.}
\newblock \bibinfo{title}{Estimating clustering coefficients and size of social
  networks via random walk}.
\newblock In \emph{\bibinfo{booktitle}{Proceedings of the 22nd International
  Conference on World Wide Web}}, \bibinfo{pages}{539--550}
  (\bibinfo{year}{2013}).

\bibitem{gjoka2013}
\bibinfo{author}{Gjoka, M.}, \bibinfo{author}{Kurant, M.} \&
  \bibinfo{author}{Markopoulou, A.}
\newblock \bibinfo{title}{{2.5K}-graphs: {From} sampling to generation}.
\newblock In \emph{\bibinfo{booktitle}{2013 Proceedings IEEE INFOCOM}},
  \bibinfo{pages}{1968--1976} (\bibinfo{year}{2013}).

\bibitem{dasgupta2014}
\bibinfo{author}{Dasgupta, A.}, \bibinfo{author}{Kumar, R.} \&
  \bibinfo{author}{Sarlos, T.}
\newblock \bibinfo{title}{{On Estimating the Average Degree}}.
\newblock In \emph{\bibinfo{booktitle}{Proceedings of the 23rd International
  Conference on World Wide Web}}, \bibinfo{pages}{795--806}
  (\bibinfo{year}{2014}).

\bibitem{wang2014}
\bibinfo{author}{Wang, P.} \emph{et~al.}
\newblock \bibinfo{journal}{\bibinfo{title}{{Efficiently Estimating Motif
  Statistics of Large Networks}}}.
\newblock {\emph{\JournalTitle{ACM Transactions on Knowledge Discovery from
  Data}}} \textbf{\bibinfo{volume}{9}} (\bibinfo{year}{2014}).
\newblock \bibinfo{note}{Article No. 8}.

\bibitem{han2016}
\bibinfo{author}{Han, G.} \& \bibinfo{author}{Sethu, H.}
\newblock \bibinfo{title}{Waddling random walk: Fast and accurate mining of
  motif statistics in large graphs}.
\newblock In \emph{\bibinfo{booktitle}{2016 IEEE 16th International Conference
  on Data Mining (ICDM)}}, \bibinfo{pages}{181--190} (\bibinfo{year}{2016}).

\bibitem{chen2016}
\bibinfo{author}{Chen, X.}, \bibinfo{author}{Li, Y.}, \bibinfo{author}{Wang,
  P.} \& \bibinfo{author}{Lui, J.}
\newblock \bibinfo{journal}{\bibinfo{title}{{A General Framework for Estimating
  Graphlet Statistics via Random Walk}}}.
\newblock {\emph{\JournalTitle{Proceedings of the VLDB Endowment}}}
  \textbf{\bibinfo{volume}{10}}, \bibinfo{pages}{253--264}
  (\bibinfo{year}{2016}).

\bibitem{nakajima2020jip}
\bibinfo{author}{Nakajima, K.} \& \bibinfo{author}{Shudo, K.}
\newblock \bibinfo{journal}{\bibinfo{title}{Estimating high betweenness
  centrality nodes via random walk in social networks}}.
\newblock {\emph{\JournalTitle{Journal of Information Processing}}}
  \textbf{\bibinfo{volume}{28}}, \bibinfo{pages}{436--444}
  (\bibinfo{year}{2020}).

\bibitem{bera2020}
\bibinfo{author}{Bera, S.~K.} \& \bibinfo{author}{Seshadhri, C.}
\newblock \bibinfo{title}{How to count triangles, without seeing the whole
  graph}.
\newblock In \emph{\bibinfo{booktitle}{Proceedings of the 26th ACM SIGKDD
  International Conference on Knowledge Discovery \& Data Mining}},
  \bibinfo{pages}{306--316} (\bibinfo{year}{2020}).

\bibitem{catanese2011}
\bibinfo{author}{Catanese, S.~A.}, \bibinfo{author}{De~Meo, P.},
  \bibinfo{author}{Ferrara, E.}, \bibinfo{author}{Fiumara, G.} \&
  \bibinfo{author}{Provetti, A.}
\newblock \bibinfo{title}{{Crawling Facebook for Social Network Analysis
  Purposes}}.
\newblock In \emph{\bibinfo{booktitle}{Proceedings of the International
  Conference on Web Intelligence, Mining and Semantics}}
  (\bibinfo{year}{2011}).
\newblock \bibinfo{note}{Article No. 52}.

\bibitem{takac2012}
\bibinfo{author}{Takac, L.} \& \bibinfo{author}{Zabovsky, M.}
\newblock \bibinfo{title}{{Data analysis in public social networks}}.
\newblock In \emph{\bibinfo{booktitle}{International Scientific Conference and
  International Workshop Present Day Trends of Innovations}}
  (\bibinfo{year}{2012}).

\bibitem{kurant}
\bibinfo{author}{Kurant, M.}, \bibinfo{author}{Gjoka, M.},
  \bibinfo{author}{Butts, C.~T.} \& \bibinfo{author}{Markopoulou, A.}
\newblock \bibinfo{title}{{Walking on a Graph with a Magnifying Glass:
  Stratified Sampling via Weighted Random Walks}}.
\newblock In \emph{\bibinfo{booktitle}{Proceedings of the ACM SIGMETRICS Joint
  International Conference on Measurement and Modeling of Computer Systems}},
  \bibinfo{pages}{281--292} (\bibinfo{year}{2011}).

\bibitem{ribeiro2010}
\bibinfo{author}{Ribeiro, B.} \& \bibinfo{author}{Towsley, D.}
\newblock \bibinfo{title}{{Estimating and Sampling Graphs with Multidimensional
  Random Walks}}.
\newblock In \emph{\bibinfo{booktitle}{Proceedings of the 10th ACM SIGCOMM
  Conference on Internet Measurement}}, \bibinfo{pages}{390--403}
  (\bibinfo{year}{2010}).

\bibitem{lee2012}
\bibinfo{author}{Lee, C.-H.}, \bibinfo{author}{Xu, X.} \& \bibinfo{author}{Eun,
  D.~Y.}
\newblock \bibinfo{journal}{\bibinfo{title}{Beyond random walk and
  metropolis-hastings samplers: Why you should not backtrack for unbiased graph
  sampling}}.
\newblock {\emph{\JournalTitle{Proceedings of the 12th ACM
  SIGMETRICS/PERFORMANCE Joint International Conference on Measurement and
  Modeling of Computer Systems}}} \textbf{\bibinfo{volume}{40}},
  \bibinfo{pages}{319–330} (\bibinfo{year}{2012}).

\bibitem{li2015}
\bibinfo{author}{Li, R.-H.}, \bibinfo{author}{Yu, J.~X.}, \bibinfo{author}{Qin,
  L.}, \bibinfo{author}{Mao, R.} \& \bibinfo{author}{Jin, T.}
\newblock \bibinfo{title}{On random walk based graph sampling}.
\newblock In \emph{\bibinfo{booktitle}{2015 IEEE 31st International Conference
  on Data Engineering}}, \bibinfo{pages}{927--938} (\bibinfo{year}{2015}).

\bibitem{nazi2015}
\bibinfo{author}{Nazi, A.}, \bibinfo{author}{Zhou, Z.},
  \bibinfo{author}{Thirumuruganathan, S.}, \bibinfo{author}{Zhang, N.} \&
  \bibinfo{author}{Das, G.}
\newblock \bibinfo{journal}{\bibinfo{title}{Walk, not wait: Faster sampling
  over online social networks}}.
\newblock {\emph{\JournalTitle{Proceedings of the VLDB Endowment}}}
  \textbf{\bibinfo{volume}{8}}, \bibinfo{pages}{678--689}
  (\bibinfo{year}{2015}).

\bibitem{zhou2016}
\bibinfo{author}{Zhou, Z.}, \bibinfo{author}{Zhang, N.}, \bibinfo{author}{Gong,
  Z.} \& \bibinfo{author}{Das, G.}
\newblock \bibinfo{journal}{\bibinfo{title}{Faster random walks by rewiring
  online social networks on-the-fly}}.
\newblock {\emph{\JournalTitle{ACM Transactions on Database Systems}}}
  \textbf{\bibinfo{volume}{40}} (\bibinfo{year}{2016}).
\newblock \bibinfo{note}{Article No. 26}.

\bibitem{li2019}
\bibinfo{author}{Li, Y.} \emph{et~al.}
\newblock \bibinfo{title}{Walking with perception: Efficient random walk
  sampling via common neighbor awareness}.
\newblock In \emph{\bibinfo{booktitle}{2019 IEEE 35th International Conference
  on Data Engineering (ICDE)}}, \bibinfo{pages}{962--973}
  (\bibinfo{year}{2019}).

\bibitem{yi2021}
\bibinfo{author}{Yi, P.}, \bibinfo{author}{Xie, H.}, \bibinfo{author}{Li, Y.}
  \& \bibinfo{author}{Lui, J.~C.}
\newblock \bibinfo{title}{A bootstrapping approach to optimize random walk
  based statistical estimation over graphs}.
\newblock In \emph{\bibinfo{booktitle}{2021 IEEE 37th International Conference
  on Data Engineering (ICDE)}}, \bibinfo{pages}{900--911}
  (\bibinfo{year}{2021}).

\bibitem{heckathorn1997}
\bibinfo{author}{Heckathorn, D.~D.}
\newblock \bibinfo{journal}{\bibinfo{title}{Respondent-driven sampling: {A} new
  approach to the study of hidden populations}}.
\newblock {\emph{\JournalTitle{Social Problems}}}
  \textbf{\bibinfo{volume}{44}}, \bibinfo{pages}{174--199}
  (\bibinfo{year}{1997}).

\bibitem{salganik2004}
\bibinfo{author}{Salganik, M.~J.} \& \bibinfo{author}{Heckathorn, D.~D.}
\newblock \bibinfo{journal}{\bibinfo{title}{Sampling and estimation in hidden
  populations using respondent-driven sampling}}.
\newblock {\emph{\JournalTitle{Sociological Methodology}}}
  \textbf{\bibinfo{volume}{34}}, \bibinfo{pages}{193--240}
  (\bibinfo{year}{2004}).

\bibitem{volz2008}
\bibinfo{author}{Volz, E.} \& \bibinfo{author}{Heckathorn, D.~D.}
\newblock \bibinfo{journal}{\bibinfo{title}{Probability based estimation theory
  for respondent driven sampling}}.
\newblock {\emph{\JournalTitle{Journal of Official Statistics}}}
  \textbf{\bibinfo{volume}{24}}, \bibinfo{pages}{79--97}
  (\bibinfo{year}{2008}).

\bibitem{robins2004}
\bibinfo{author}{Robins, G.}, \bibinfo{author}{Pattison, P.} \&
  \bibinfo{author}{Woolcock, J.}
\newblock \bibinfo{journal}{\bibinfo{title}{Missing data in networks:
  exponential random graph (p*) models for networks with non-respondents}}.
\newblock {\emph{\JournalTitle{Social Networks}}}
  \textbf{\bibinfo{volume}{26}}, \bibinfo{pages}{257--283}
  (\bibinfo{year}{2004}).

\bibitem{goel2010}
\bibinfo{author}{Goel, S.} \& \bibinfo{author}{Salganik, M.~J.}
\newblock \bibinfo{journal}{\bibinfo{title}{Assessing respondent-driven
  sampling}}.
\newblock {\emph{\JournalTitle{Proceedings of the National Academy of
  Sciences}}} \textbf{\bibinfo{volume}{107}}, \bibinfo{pages}{6743--6747}
  (\bibinfo{year}{2010}).

\bibitem{illenberger2012}
\bibinfo{author}{Illenberger, J.} \& \bibinfo{author}{Flötteröd, G.}
\newblock \bibinfo{journal}{\bibinfo{title}{Estimating network properties from
  snowball sampled data}}.
\newblock {\emph{\JournalTitle{Social Networks}}}
  \textbf{\bibinfo{volume}{34}}, \bibinfo{pages}{701--711}
  (\bibinfo{year}{2012}).

\bibitem{gile2015}
\bibinfo{author}{Gile, K.~J.}, \bibinfo{author}{Johnston, L.~G.} \&
  \bibinfo{author}{Salganik, M.~J.}
\newblock \bibinfo{journal}{\bibinfo{title}{Diagnostics for respondent-driven
  sampling}}.
\newblock {\emph{\JournalTitle{Journal of the Royal Statistical Society: Series
  A (Statistics in Society)}}} \textbf{\bibinfo{volume}{178}}
  (\bibinfo{year}{2015}).
\newblock \bibinfo{note}{Article No. 241}.

\bibitem{western2016}
\bibinfo{author}{Western, B.}, \bibinfo{author}{Braga, A.},
  \bibinfo{author}{Hureau, D.} \& \bibinfo{author}{Sirois, C.}
\newblock \bibinfo{journal}{\bibinfo{title}{Study retention as bias reduction
  in a hard-to-reach population}}.
\newblock {\emph{\JournalTitle{Proceedings of the National Academy of
  Sciences}}} \textbf{\bibinfo{volume}{113}}, \bibinfo{pages}{5477--5485}
  (\bibinfo{year}{2016}).

\bibitem{tomas2011}
\bibinfo{author}{Tomas, A.} \& \bibinfo{author}{Gile, K.~J.}
\newblock \bibinfo{journal}{\bibinfo{title}{The effect of differential
  recruitment, non-response and non-recruitment on estimators for
  respondent-driven sampling}}.
\newblock {\emph{\JournalTitle{Electronic Journal of Statistics}}}
  \textbf{\bibinfo{volume}{5}}, \bibinfo{pages}{899--934}
  (\bibinfo{year}{2011}).

\bibitem{lu2012}
\bibinfo{author}{Lu, X.} \emph{et~al.}
\newblock \bibinfo{journal}{\bibinfo{title}{The sensitivity of
  respondent-driven sampling}}.
\newblock {\emph{\JournalTitle{Journal of the Royal Statistical Society: Series
  A (Statistics in Society)}}} \textbf{\bibinfo{volume}{175}},
  \bibinfo{pages}{191--216} (\bibinfo{year}{2012}).

\bibitem{Rocha2017}
\bibinfo{author}{Rocha, L. E.~C.}, \bibinfo{author}{Thorson, A.~E.},
  \bibinfo{author}{Lambiotte, R.} \& \bibinfo{author}{Liljeros, F.}
\newblock \bibinfo{journal}{\bibinfo{title}{Respondent-driven sampling bias
  induced by community structure and response rates in social networks}}.
\newblock {\emph{\JournalTitle{Journal of the Royal Statistical Society: Series
  A (Statistics in Society)}}} \textbf{\bibinfo{volume}{180}},
  \bibinfo{pages}{99--118} (\bibinfo{year}{2017}).

\bibitem{albert2000}
\bibinfo{author}{Albert, R.}, \bibinfo{author}{Jeong, H.} \&
  \bibinfo{author}{Barab{\'a}si, A.-L.}
\newblock \bibinfo{journal}{\bibinfo{title}{{Error and attack tolerance of
  complex networks}}}.
\newblock {\emph{\JournalTitle{Nature}}} \textbf{\bibinfo{volume}{406}},
  \bibinfo{pages}{378--382} (\bibinfo{year}{2000}).

\bibitem{kossinets2006}
\bibinfo{author}{Kossinets, G.}
\newblock \bibinfo{journal}{\bibinfo{title}{{Effects of missing data in social
  networks}}}.
\newblock {\emph{\JournalTitle{Social Networks}}}
  \textbf{\bibinfo{volume}{28}}, \bibinfo{pages}{247--268}
  (\bibinfo{year}{2006}).

\bibitem{huisman2009}
\bibinfo{author}{Huisman, M.}
\newblock \bibinfo{journal}{\bibinfo{title}{Imputation of missing network data:
  Some simple procedures}}.
\newblock {\emph{\JournalTitle{Journal of Social Structure}}}
  \textbf{\bibinfo{volume}{10}}, \bibinfo{pages}{1--29} (\bibinfo{year}{2009}).

\bibitem{smith2013}
\bibinfo{author}{Smith, J.~A.} \& \bibinfo{author}{Moody, J.}
\newblock \bibinfo{journal}{\bibinfo{title}{Structural effects of network
  sampling coverage {I}: {N}odes missing at random}}.
\newblock {\emph{\JournalTitle{Social Networks}}}
  \textbf{\bibinfo{volume}{35}}, \bibinfo{pages}{652--668}
  (\bibinfo{year}{2013}).

\bibitem{cohen2000}
\bibinfo{author}{Cohen, R.}, \bibinfo{author}{Erez, K.},
  \bibinfo{author}{Ben-Avraham, D.} \& \bibinfo{author}{Havlin, S.}
\newblock \bibinfo{journal}{\bibinfo{title}{Resilience of the internet to
  random breakdowns}}.
\newblock {\emph{\JournalTitle{Physical Review Letters}}}
  \textbf{\bibinfo{volume}{85}} (\bibinfo{year}{2000}).
\newblock \bibinfo{note}{Article No. 4626}.

\bibitem{nakajima2021scirep}
\bibinfo{author}{Nakajima, K.} \& \bibinfo{author}{Shudo, K.}
\newblock \bibinfo{journal}{\bibinfo{title}{Measurement error of network
  clustering coefficients under randomly missing nodes}}.
\newblock {\emph{\JournalTitle{Scientific Reports}}}
  \textbf{\bibinfo{volume}{11}} (\bibinfo{year}{2021}).
\newblock \bibinfo{note}{Article No. 2815}.

\bibitem{nakajima2020kdd}
\bibinfo{author}{Nakajima, K.} \& \bibinfo{author}{Shudo, K.}
\newblock \bibinfo{title}{Estimating properties of social networks via random
  walk considering private nodes}.
\newblock In \emph{\bibinfo{booktitle}{Proceedings of the 26th ACM SIGKDD
  International Conference on Knowledge Discovery \& Data Mining}},
  \bibinfo{pages}{720--730} (\bibinfo{year}{2020}).

\bibitem{twitter_api_followers}
\bibinfo{author}{Twitter}.
\newblock \bibinfo{title}{{Twitter API GET followers/ids}}.
\newblock
  \bibinfo{howpublished}{\url{https://developer.twitter.com/en/docs/accounts-and-users/follow-search-get-users/api-reference/get-followers-ids.html}}
  (\bibinfo{year}{2021}).
\newblock \bibinfo{note}{Accessed on December 2021}.

\bibitem{twitter_api_friends}
\bibinfo{author}{Twitter}.
\newblock \bibinfo{title}{Twitter api get friends/ids}.
\newblock
  \bibinfo{howpublished}{\url{https://developer.twitter.com/en/docs/accounts-and-users/follow-search-get-users/api-reference/get-friends-ids.html}}
  (\bibinfo{year}{2021}).
\newblock \bibinfo{note}{Accessed on December 2021}.

\bibitem{chiericetti}
\bibinfo{author}{Chiericetti, F.}, \bibinfo{author}{Dasgupta, A.},
  \bibinfo{author}{Kumar, R.}, \bibinfo{author}{Lattanzi, S.} \&
  \bibinfo{author}{Sarl{\'o}s, T.}
\newblock \bibinfo{title}{On sampling nodes in a network}.
\newblock In \emph{\bibinfo{booktitle}{Proceedings of the 25th International
  Conference on World Wide Web}}, \bibinfo{pages}{471--481}
  (\bibinfo{year}{2016}).

\bibitem{levin2017}
\bibinfo{author}{Levin, D.~A.} \& \bibinfo{author}{Peres, Y.}
\newblock \emph{\bibinfo{title}{{Markov Chains and Mixing Times}}}, vol.
  \bibinfo{volume}{107} (\bibinfo{publisher}{American Mathematical Soc.},
  \bibinfo{year}{2017}).

\bibitem{jones}
\bibinfo{author}{Jones, G.~L.}
\newblock \bibinfo{journal}{\bibinfo{title}{{On the Markov chain central limit
  theorem}}}.
\newblock {\emph{\JournalTitle{Probability Surveys}}}
  \textbf{\bibinfo{volume}{1}}, \bibinfo{pages}{299--320}
  (\bibinfo{year}{2004}).

\bibitem{roberts}
\bibinfo{author}{Roberts, G.~O.} \& \bibinfo{author}{Rosenthal, J.~S.}
\newblock \bibinfo{journal}{\bibinfo{title}{{General state space Markov chains
  and MCMC algorithms}}}.
\newblock {\emph{\JournalTitle{Probability Surveys}}}
  \textbf{\bibinfo{volume}{1}}, \bibinfo{pages}{20--71} (\bibinfo{year}{2004}).

\bibitem{Lovasz1996}
\bibinfo{author}{Lov\'asz, L.}
\newblock \bibinfo{title}{Random walks on graphs: {A} survey}.
\newblock In \emph{\bibinfo{booktitle}{Combinatorics, Paul Erd\H{o}s is
  Eighty}}, vol.~\bibinfo{volume}{2}, \bibinfo{pages}{353--398}
  (\bibinfo{year}{1996}).

\bibitem{newman_networks}
\bibinfo{author}{Newman, M. E.~J.}
\newblock \emph{\bibinfo{title}{Networks. Second Edition}}
  (\bibinfo{publisher}{Oxford University Press}, \bibinfo{address}{Oxford, UK},
  \bibinfo{year}{2018}).

\bibitem{watts1998}
\bibinfo{author}{Watts, D.~J.} \& \bibinfo{author}{Strogatz, S.~H.}
\newblock \bibinfo{journal}{\bibinfo{title}{Collective dynamics of
  `small-world'networks}}.
\newblock {\emph{\JournalTitle{Nature}}} \textbf{\bibinfo{volume}{393}},
  \bibinfo{pages}{440--442} (\bibinfo{year}{1998}).

\bibitem{konect}
\bibinfo{author}{Kunegis, J.}
\newblock \bibinfo{title}{{KONECT–The Koblenz Network Collection}}.
\newblock In \emph{\bibinfo{booktitle}{Proceedings of the 22nd International
  Conference on World Wide Web}}, \bibinfo{pages}{1343--1350}
  (\bibinfo{year}{2013}).

\bibitem{snap}
\bibinfo{author}{Leskovec, J.} \& \bibinfo{author}{Krevl, A.}
\newblock \bibinfo{title}{{SNAP Datasets}: {Stanford} large network dataset
  collection}.
\newblock \bibinfo{howpublished}{\url{http://snap.stanford.edu/data}}
  (\bibinfo{year}{2014}).

\bibitem{nr}
\bibinfo{author}{Rossi, R.~A.} \& \bibinfo{author}{Ahmed, N.~K.}
\newblock \bibinfo{title}{{The Network Data Repository with Interactive Graph
  Analytics and Visualization}}.
\newblock In \emph{\bibinfo{booktitle}{Proceedings of the Twenty-Ninth AAAI
  Conference on Artificial Intelligence}}, \bibinfo{pages}{4292--4293}
  (\bibinfo{year}{2015}).

\bibitem{dey2012}
\bibinfo{author}{Dey, R.}, \bibinfo{author}{Jelveh, Z.} \&
  \bibinfo{author}{Ross, K.}
\newblock \bibinfo{title}{Facebook users have become much more private: A
  large-scale study}.
\newblock In \emph{\bibinfo{booktitle}{2012 IEEE International Conference on
  Pervasive Computing and Communications Workshops}}, \bibinfo{pages}{346--352}
  (\bibinfo{year}{2012}).

\bibitem{buccafurri2015}
\bibinfo{author}{Buccafurri, F.}, \bibinfo{author}{Lax, G.},
  \bibinfo{author}{Nicolazzo, S.} \& \bibinfo{author}{Nocera, A.}
\newblock \bibinfo{journal}{\bibinfo{title}{Comparing twitter and facebook user
  behavior: Privacy and other aspects}}.
\newblock {\emph{\JournalTitle{Computers in Human Behavior}}}
  \textbf{\bibinfo{volume}{52}}, \bibinfo{pages}{87--95}
  (\bibinfo{year}{2015}).

\bibitem{facebook_500m}
\bibinfo{author}{Facebook}.
\newblock \bibinfo{title}{{Facebook 500 Million Stories}}.
\newblock
  \bibinfo{howpublished}{\url{https://www.facebook.com/notes/facebook/500-million-stories/409753352130/}}
  (\bibinfo{year}{2010}).

\bibitem{erdos1959}
\bibinfo{author}{Erd{\H{o}}s, P.} \& \bibinfo{author}{R{\'e}nyi, A.}
\newblock \bibinfo{journal}{\bibinfo{title}{On random graphs {I}}}.
\newblock {\emph{\JournalTitle{Publ. Math. Debrecen}}}
  \textbf{\bibinfo{volume}{6}}, \bibinfo{pages}{290--297}
  (\bibinfo{year}{1959}).

\bibitem{barabasi1999}
\bibinfo{author}{Barabási, A.-L.} \& \bibinfo{author}{Albert, R.}
\newblock \bibinfo{journal}{\bibinfo{title}{Emergence of scaling in random
  networks}}.
\newblock {\emph{\JournalTitle{Science}}} \textbf{\bibinfo{volume}{286}},
  \bibinfo{pages}{509--512} (\bibinfo{year}{1999}).

\end{thebibliography}

\end{document}